\documentclass[12pt, draftclsnofoot, onecolumn]{IEEEtran}
\usepackage{amsmath}
\usepackage[pdftex]{graphicx}
\usepackage{amsthm}
\usepackage{amsfonts}
\usepackage{amssymb}
\usepackage{epstopdf}
\usepackage{cite}
\usepackage{subfigure}
\usepackage{color}

\newtheorem{MyLem}{Lemma}
\newtheorem{MyTheo}{Theorem}
\newtheorem{MyProp}{Proposition}
\newtheorem{MyCoro}{Corollary}

\begin{document}
\title{\huge{On the Performance of Data Compression in Clustered Fog Radio Access Networks}}
\author{\IEEEauthorblockN{
Haonan Hu, \textit{Member IEEE},
Yan Jiang,
Jiliang Zhang, \textit{Senior Member IEEE},
Yanan Zheng,
Qianbin Chen \textit{Senior Member IEEE},
and Jie Zhang \textit{Senior Member IEEE}}\\
\thanks{This work was supported in part by the National Natural
Science Foundation of China (NSFC) under the grant number
61901075, 61831002, in part by the Natural Science Foundation of Chongqing, China, under the grant number cstc2019jcyj-msxmX0602, in part by the National Key R$\&$D Program of China under the grant number 2017YFE0118900, and in part by the EC H2020 project is3DMIMO under the grant number 734798. 

H. Hu, Y. Zheng and Q. Chen are with the School of Communication and Information Engineering, Chongqing University of Posts and Telecommunications, Chongqing 400065, China.

Y. Jiang, J.Zhang, and J. Zhang are with the Department of Electronic and Electrical Engineering, University of Sheffield, Sheffield
S1 3JD, U.K.(e-mail: yjiang71@sheffield.ac.uk)}
}
\maketitle

\begin{abstract}
The fog-radio-access-network (F-RAN) has been proposed to address the strict latency requirements, which offloads computation tasks generated in user equipments (UEs) to the edge to reduce the processing latency. However, it incorporates the task transmission latency, which may become the bottleneck of latency requirements. Data compression (DC) has been considered as one of the promising techniques to reduce the transmission latency. By compressing the computation tasks before transmitting, the transmission delay is reduced due to the shrink transmitted data size, and the original computing task can be retrieved by employing data decompressing (DD) at the edge nodes or the centre cloud. Nevertheless, the DC and DD incorporate extra processing latency, and the latency performance has not been investigated in the large-scale DC-enabled F-RAN. Therefore, in this work, the successful data compression probability (SDCP), i.e., the probability of the task execution latency being smaller than a target latency and the signal to interference ratio of the wireless uplink transmission being larger than a threshold, is defined to analyse the latency performance of the F-RAN. Moreover, to analyse the effect of compression offloading ratio (COR), which determines the proportion of tasks being compressed at the edge, on the SDCP of the F-RAN, a novel hybrid compression mode is proposed based on the queueing theory. Based on this, the closed-form result of SDCP in the large-scale DC-enabled F-RAN is derived by combining the Matern cluster process and M/G/1 queueing model, and validated by Monte Carlo simulations. Based on the derived SDCP results, the effects of COR on the SDCP is analysed numerically. The results show that the SDCP with the optimal COR can be enhanced with a maximum value of $0.3$ and $0.55$ as compared with the cases of compressing all computing tasks at the edge and at the UE, respectively. Moreover, for the system requiring the minimal latency, the proposed hybrid compression mode can alleviate the requirement on the backhaul capacity. 
\end{abstract}

\section{Introduction}
With the unprecedented proliferation of new service types, the ultra reliable low-latency communications (URLLC) has been considered as a key requirement for the fifth generation (5G) networks to support latency-sensitive services \cite{tullberg2016the}, e.g., the virtual reality (VR) and the augmented reality (AR). This latency requirement has been further improved in the future sixth generation (6G) networks, which claims that the experienced end-to-end latency should be in the range of $0.01$-$0.1$ ms \cite{zhang20196g}. To meet the latency requirement, the fog-radio-access network (F-RAN) has been advocated \cite{peng2015contract} as a new network architecture, which utilizes the processing capacities in the UEs and edge nodes. 

The F-RAN is composed of three layers: the terminal layer, the access layer and the cloud computing layer \cite{peng2015contract}. The terminal layer consists of the user equipments (UEs), and the access layer consists of fog nodes (FNs) and fog access points (FAPs). These are all equipped with limited computing capabilities, and the access layer usually has higher computing capabilities as compared with the terminal layer. The cloud computing layer is composed of the cloud computing centres that equipped with powerful computing capable base-band-unit (BBU) pools. The latency of computation tasks can be significantly reduced for two major reasons. One is that the latency-sensitive tasks can be executed at the network edge, and the other is that by deploying fibre-optical fronthaul links between the FAP and the center cloud, the unstable transmission delay in the core network can be avoided. However, to support the flexibility of F-RAN deployment, the wireless backhaul may be adopted to link the FN to its associated FAP \cite{liu2018cooperative}. As a result, the capacity of wireless backhaul is constrained, which may become the bottleneck for the latency requirement in the F-RAN. 

To address this, the data compression (DC) for the computation tasks has been advocated in the F-RAN. By compressing the data of computation tasks before transmitting on the backhaul, the data size shrinks thus the transmission delay can be reduced. After the compressed computation task receiving in cloud, the original computation task can be retrieved by data decompression (DD). The DC schemes are usually classified into two categories: lossy and lossless DC algorithms. On one hand, the lossy DC algorithms are mainly used for the digitally sampled analog data, which mostly includes sound, video, graphics, or picture files. This is because the original data cannot be recovered exactly the same as it was compressed by lossy algorithms before, and a relative small information loss is acceptable for such data types. On the other hand, the lossless DC algorithms are mainly used for textual data, which can fully recover the original data by decoding the encoded repeated pattern in it. Nevertheless, employing the DC and DD in the F-RAN will introduce additional computing latency. Thus, the impact of DC and DD on the latency performance of F-RAN should be investigated. 


There have been some existing works studying the DC. In early works \cite{tavli2010optimal,deepu2017a,le2018lossless}, the DC was mainly discussed in the wireless sensor networks (WSNs), which aimed to design new DC techniques or schemes to save the energy consumption of sensors while ensuring lossless transmission in the network. However, the delay of DC for real-time applications was ignored. To reduce the additional latency introduced by the DC for the WSNs, a novel DC schemes utilizing Huffman trees was
proposed in \cite{gia2019lossless}, which is applicable for a wide range of real-time systems. Recently, the DC technique have been investigated in the network with mobile edge computing (MEC). In \cite{ly2019joint}, the energy consumption in the MEC-enabled network was optimized by jointly adjusting the task allocation decision and the compression ratio. The results demonstrated that the proposed scheme outperforms that without DC before transmission. Motivated by the computation offloading \cite{guo2018computation,du2018computation,liu2018multiobjective}, whether and where to compress the data in the MEC-enabled networks was investigated in \cite{ren2019data}. Specifically, the end-to-end latency was derived with DC respectively operating on the UE and the edge node, and the optimal compression ratio was analysed to minimize the end-to-end latency. The results showed that this optimal value has a binary structure and is mainly affected by the UE computation capacity and the bandwidth of wireless transmission. In \cite{xu2019energy}, the energy consumption of the MEC-enabled network was minimized by jointly optimizing the computation offloading, DC and transmission duration allocation with latency constraint. The results showed that the proposed algorithm outperforms the algorithms with all data operating DC and without DC in terms of the energy consumption. Based on these, work \cite{Nguyen2020joint} proposed a joint optimization algorithm to minimize a weighted sum of the energy consumption and end-to-end delay of all users by computation offloading, radio and computing resources allocation, and compression ratio adjustment. The results showed that the proposed optimal algorithm for DC at UEs and FNs can significantly reduce the weighted sum of energy consumption and service delay as compared with that do not leverage DC. However, the queueing delays in the DC, DD and task computation were ignored in all the above mentioned works. In \cite{liu2018cooperative}, the service outage probability was analysed in the F-RAN with DC. The service outage probability is defined as either the signal to interference plus ratio (SINR) of the wireless link being smaller than a threshold or the computation latency being larger than a threshold. Combined with the queueing theory, the latency of DC at the FN and that of DD and task computing at the FAP were modelled as two independent M/M/1 queue. The results showed that the F-RAN with DC outperforms that without DC in terms of the service outage probability. To our best knowledge, \cite{liu2018cooperative} is the only work utilizing the queueing theory to model the DC, DD and task computation, but the M/M/1 queue is not suitable for the latency modelling of DD and task computing at the FAP. This is due to that the DD and task computing include two independent exponential-distributed service time, while the M/M/1 queue assumes that the service time follows the exponential distribution. Moreover, the effect of COR on the latency performance is ignored in this work.
Additionally, all the above mentioned works considered a small-scale F-RAN (MEC-enabled network), which consists of limited number of FNs and FAPs (edge nodes and edge clouds). Consequently, the conclusions or design insights obtained in these works may not be valid for a large-scale F-RAN for the latency analysis.  

Stochastic geometry has been proven to be powerful mathematical tools to model and analyse the performance of large-scale networks \cite{andrews2011tractable,hu2018coverage,liu2020the,hu2020on} and there are only a few works using stochastic geometry analysing the performance of large-scale randomly deployed 
MEC-enabled cellular networks or F-RAN. In \cite{kong2018fog}, the downlink coverage probability in a device-to-device enabled F-RAN was analysed, where the positions of FAPs were modelled following the Ginibre point process (GPP). The results indicated that the repulsion of FAPs introduced by GPP led to a higher coverage probability as compared with that obtained without repulsion following the Poisson point process (PPP) \cite{hu1016coverage}. However, the results in \cite{kong2018fog} cannot be applied in the latency performance of DC-enabled F-RAN as this latency performance is mainly affected by the uplink transmission performance. The latency performance in the large-scale F-RAN was mainly investigated for the computation offloading, which aims to reduce the end-to-end latency by offloading intensive computation tasks from UEs to edge nodes. In \cite{ko2018wireless}, the end-to-end latency was investigated in a large-scale MEC-enabled cellular networks generated following PPP under synchronous and asynchronous offloading cases. In these two cases, the tasks respectively arrives randomly over time and arrives periodically at the beginning of each time slot. The results showed that the computing latency in the synchronous case is longer than that in the asynchronous case, which leads to a longer end-to-end executing latency. In \cite{park2018successful}, the successful edge computing probability, which was defined as the probability of total latency being less than the target latency, was analysed in the MEC-enabled heterogeneous networks. The results indicated that the successful edge computing probability can be increased by offloading more tasks to the high-computing-capability MEC sever. Nevertheless, this work ignored that the UEs have clustered nature around their associated FNs in the F-RAN. Accordingly, \cite{xu2019performance} modelled the positions of FNs and UEs following the Thomas cluster process. Based on this, the average delay of the large-scale F-RAN with computation offloading was analysed by using the M/M/1 queue to model the task execution delay. An optimal ratio of computing tasks processed by cloud centre was given and the results showed that this optimal ratio can reduce the average delay significantly as compared with offloading the computing tasks only on the FN or the cloud centre. However, all these works did not take the DC into consideration, which has significant impact on the processing delay modelling and the transmission delay, especially in the F-RAN with constrained backhaul. We summarize the motivations of our work as the following three points:

1) The existing latency models of DD and task computing either ignore the queueing delay or simply use the M/M/1 queue, which is impractical in a actual. Therefore, a more practical model for DD and task computing is required for latency analysis.

2) The existing works that investigated the network latency performance of a large-scale spatially distributed F-RAN, which consists of large numbers of randomly deployed FNs, mainly considered the computation offloading without DC. The latency performance of the DC-enabled F-RAN, where DC can be fully or partially operated on the UE or the FNs, has been ignored so far, which can provide system design insights on the actual DC-enabled F-RAN deployment. 

3) In the literature, the effect of backhaul capacity on the latency performance in the DC-enabled F-RAN has not been investigated yet.

To our best knowledge, we are the first to propose a latency model for the DD and task computation based on the M/G/1 queueing model. Moreover, we propose a novel hybrid compression mode to analyse the effect of compression offloading ratio (COR), which determines the proportion of tasks being compressed at the edge, on the latency performance of the F-RAN.    
Equipped with this, we investigate the successful data compression probability (SDCP) in the large-scale clustered F-RAN under DC fully or partially operating on the UE or the FN. The SDCP is defined as the probability of the execution latency being less than a target latency and the signal to interference ratio (SIR) of the wireless uplink transmission being larger than a threshold. Specifically, the explicit expression of SDCP is derived in closed from and validated by Monte Carlo simulations. Based on this, the effect of the COR and backhaul capacity on the SDCP is analysed numerically. Furthermore, an optimal COR to maximize the SDCP is obtained based on the derived SDCP. Our main contributions of this work can be summarised as follows: 

1) We are the first to propose a latency model for the data decompression (DD) and task computing based on the M/G/1 queuing model. Specifically, the distribution of the latency spent on the DD and task computing is obtained, which can provide guidelines for other readers who are investigating the latency performance in networks adopting the DC. 

2) To our best knowledge, we are the first to obtain the closed-form results of the SDCP in the large-scale random deployed DC-enabled F-RAN, where the computation tasks can be partially compressed at the UEs or the edge nodes. The derived closed-form results are validated by the Monte Carlo simulations. Based on this, the optimal ratio of computation tasks being compressed at the edge nodes is analysed numerically. 

3) The results show that by appropriately allocating computation tasks being compressed at UEs and edge nodes, the SDCP can be improved significantly. Under our simulation environment, the SDCP with the optimal COR can be enhanced with a maximum of $0.3$ and $0.55$ as compared with the cases of compressing all computing tasks at the edge and at the UE, respectively. The optimal COR ranges from $0.55$ to $1$. In addition, the optimal COR decreases with the increasing of backhaul capacity, and the proposed hybrid compression mode can alleviate the requirement on the backhaul capacity.

The rest of this paper is organised as follows: Section \ref{sec_sys_mod} introduces the network model and the DC and task computing model. In Section \ref{sec_perf_analysis}, the analytical results of the SDCP in the F-RAN are given. Section \ref{sec_simu_res} presents the validation results and the analysis results before concluding the paper in Section \ref{sec_conclusion}.

\section{System Model}
\label{sec_sys_mod}
\subsection{Network Model}
\begin{figure}
	\centering
	\includegraphics[width = 0.56\linewidth]{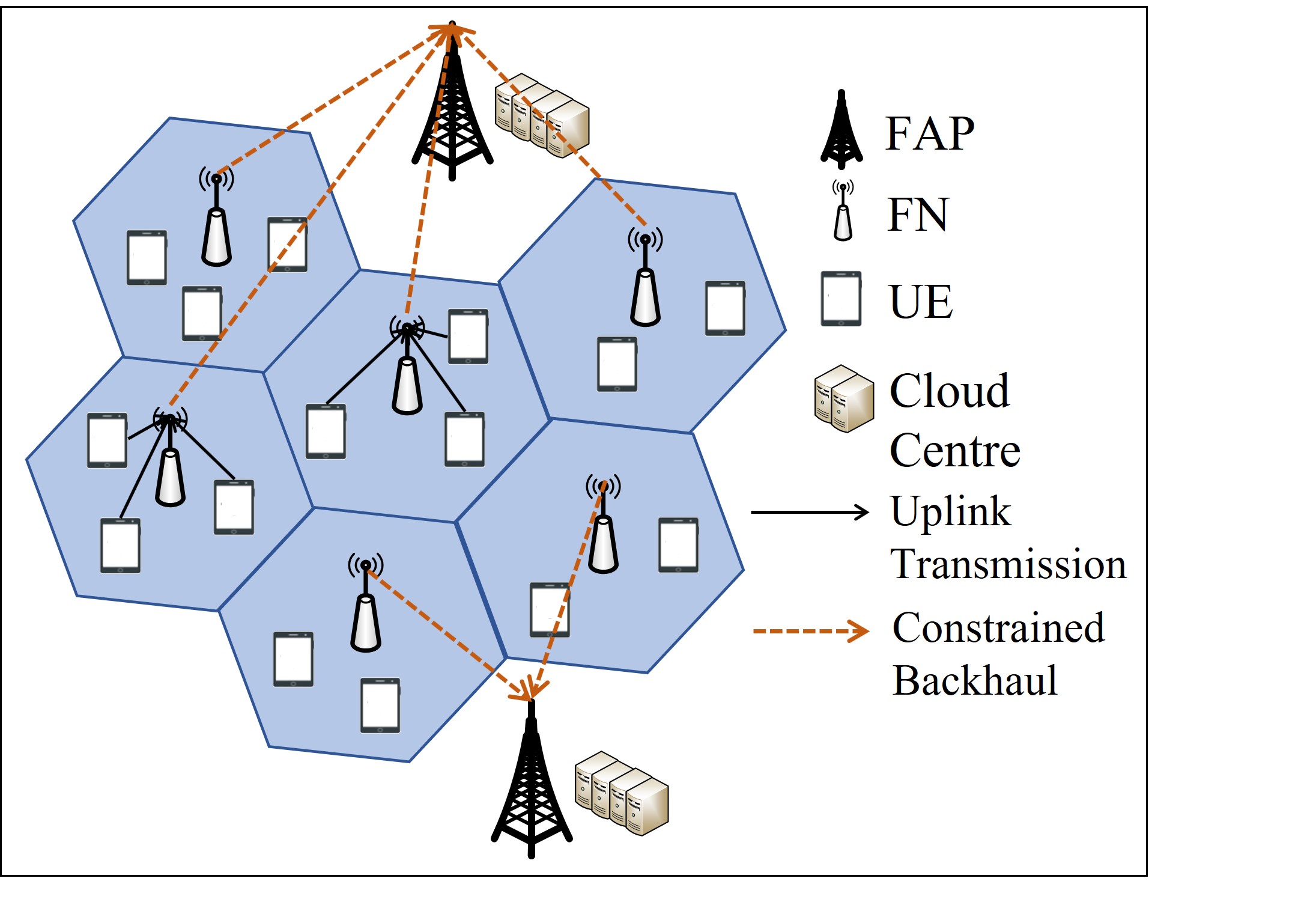}
	\caption{The illustration of system model}	
	\label{fig_sys_model}
\end{figure}
We consider a large-scale uplink F-RAN consisting of FAPs, FNs and UEs. As illustrated in Fig. \ref{fig_sys_model}, in the F-RAN, each UE is connected to its associated FN via wireless link, each FN is connected to its corresponding FAP via a constrained backhaul due to the flexibility of FN deployment \cite{liu2018cooperative}, and each FAP is connected to the cloud via an optical-fibre link. To reflect the cluster nature of UEs, we adopt the Matern cluster process (MCP) \cite{peng2019analysis} to model the positions of FNs and their corresponding UEs, where the positions of FNs are deployed following the Poisson point process (PPP) with density $\lambda_{\rm N}$, denoted by $\Phi^{\rm N}=\{X^{\rm N}_j\}$, $j\in M_{\rm N}$ with $M_{\rm N}$ being the number of FNs. The corresponding UEs are scattering in a circle region around the FN following uniform distribution with radius being $1/\sqrt{c}$ \cite{kobayashi2014uplink}, where  $c=\pi \lambda_{\rm N}$, and note that the group of UEs connecting to the same FN is assumed to use orthogonal frequency resources, thus there only exits an interfering UE in each of other FNs for a specific UE. The locations of UEs are denoted by  $Y_{i,j}$, $i\in M_j^{\rm U}$, where  $M_j^{\rm U}$ is the number of UEs in FN $j$. The positions of FAPs $\Phi^{\rm A}=\{X^{\rm A}_n\}$ are modelled following another PPP with density $\lambda_{\rm A}$, where $n\in M_{\rm A}$ with $M_{\rm A}$ being the number of FAPs, and FNs are associated with their nearest FAPs, which provide edge computing capability. To simplify the analysis, we assume that the capacity of constrained backhaul connecting the FN and the FAP is constant, denoted by ${C^{\rm bh}_{j,n}}$. 

For each link between a UE and an FN, it experiences pathloss and small-scale fading. The pathloss follows a log-distance model, which can be denoted by $l(r)=r^{-\alpha}$, where $r$ is the Euclidean distance between the transmitting and receiving nodes , and $\alpha$ is the pathloss exponent. The Rayleigh fading is assumed for the small-scale fading, thus the caused received power attenuation is modelled as an independent exponential distribution with rate parameter being $1$. We denote the small-scale fading of the link between the UE $Y_{i,j}$ and the FN $k$ as $h_{j,k}$. For a specific UE $Y_{i,j}$ associated with the FN $X^{\rm N}_j$, the uplink transmit power employs a fractional channel inversion power control mechanism given as $P_{i,j}=pl(R_{i,j})^{-\epsilon}$ \cite{kong2017modeling}, where $\epsilon$ is the power control factor, $p$ is the highest UE transmit power, and $R_{i,j}$ is the distance between the UE $Y_{i,j}$ and the FN $X^{\rm N}_j$.

Without loss of generality, we can place the target FN at the origin according to the Slyvnyak's theorem \cite{andrews2011tractable}.  Recall that there only exists an interfering UE in each of other FNs for a specific UE associated with the target FN. For denotational simplicity, the small-scale fading of the link from the associated specific UE and the interfering UE associated with FN $X^{\rm N}_j$ are denoted by $h_0$ and $h_k$, respectively. The Euclidean distance between the interfering UE and its associated FN $X^{\rm N}_k$ as $R_k$. As a result, the uplink signal to interference ratio (SIR) of the UE in the target FN can be represented as follows:
\begin{equation}
\label{eq_sir_def}
{\rm{SIR}}=\frac{ph_0l(|Y_0|)^{1-\epsilon}}{{\sum_{k \in \Phi^{\rm N}\backslash\{{0}\}}}
	ph_k\left[l(R_k)\right]^{-\epsilon}l(|Y_k|)},
\end{equation}
where the location of the specific UE associated with the target FN is denoted by $Y_0$, the location of the interfering UE associated with FN $X^{\rm N}_k$ is denoted by $Y_k$.  
Equipped with this SIR, the uplink transmission rate can be achieved by $B\log(1+\rm SIR)$, where $B$ is the subchannel bandwidth.

\subsection{Data Compression and Task Computing Model}

To alleviate the burden on the backhaul, the DC technique is applied in the F-RAN. The entire procedure of the DC can be described as follows. Firstly, the computation tasks are generated in each UE. For a specific UE $Y_{i,j}$ associated with the FN $X_j^{\rm N}$, the task generation rate is assumed to follow a Poisson process with generation rate $\psi_{i,j}$. Secondly, based on the processing capabilities of the UE and the associated FN, the UE or the FN determines whether to fully or partially compress the computation tasks. Note that it is unnecessary to employ DC at FAPs as the FAPs are connected to the cloud centre via optical-fibre links, which can support high data transmission rate. For analytical tractability, we assume that all the computation tasks are either compressed at the UE or at the FN. The case that some computation tasks may not be compressed in the UE or FN is beyond our research scope, and can be addressed by some extra work based on our work. After compressing computation tasks, the compressed data is transmitted from the FN to the FAP via the backhaul, and then is sent to the cloud centre for DD. After this, the original computation tasks are retrieved and processed in the cloud centre. In this work, we consider three DC modes based on the proportion of computing tasks being compressed at the UE, namely the local compression mode, the edge compression and the hybrid compression modes. As illustrated in Fig. \ref{fig_cm_model}, for the local and edge compression modes, all computation tasks are compressed at the UEs and the FNs, respectively. For the hybrid compression scheme, some of the computation tasks are compressed at the UEs and the rest of them are compressed at the associated FN. We assume that within the service range of an FN, there are $M_{\rm U}$ UEs that generate computation tasks that need to be uploaded or compressed. 

\begin{figure}
	\centering
	\includegraphics[width = \linewidth]{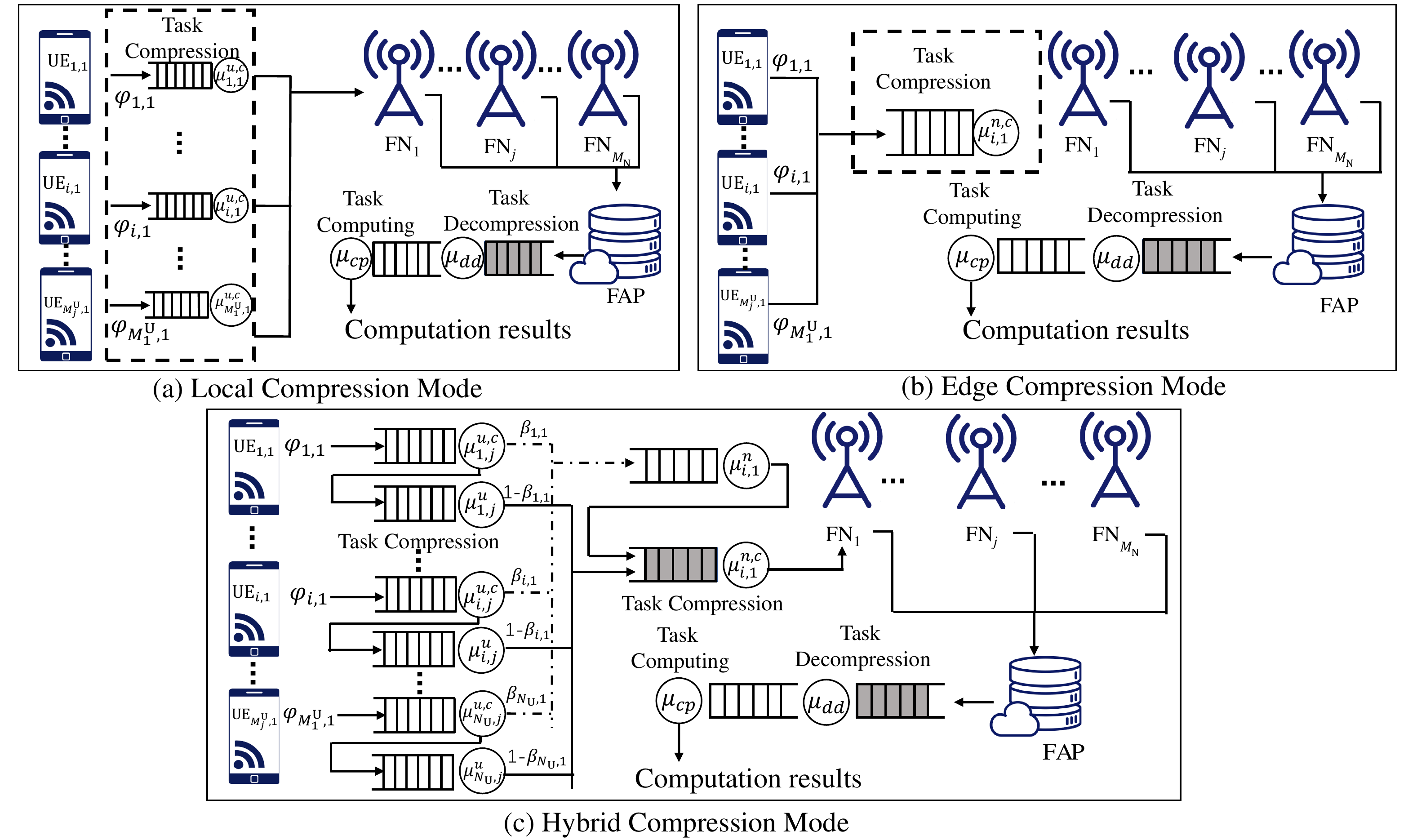}
	\caption{The illustration of three compression modes}	
	\label{fig_cm_model}
\end{figure}

The entire end-to-end delay is defined as the time duration beginning from the computation tasks generation at the UE and ending at the UE obtaining the final computing results.  According to the above discussion, the entire end-to-end latency in this work is composed by the uplink transmission delay from UE to the FN, the DC processing delay in the UE or the FN, the backhaul transmission delay from FN to the FAP, and the DD and task computing delay in the FAP. The downlink transmission delay of the computing results from the FAP to the UE is ignored due to the size of computing results usually have limited size \cite{cui2019joint}. For the aforementioned three DC modes, the backhaul transmission delay $T^{\rm bh}_{i,j}$ and the DD and task computing delay $T^{\rm dp}_{i,j,n}$ can be modelled in general forms as follows. On one hand, recall that the backhaul capacity of FN $X_j^{\rm N}$ to the $X_n^{\rm A}$ is denoted by $C^{\rm bh}_{j,n}$, the backhaul transmission delay can be calculated by
\begin{equation}
\label{eq_delay_bh}
T^{\rm bh}_{j,n}=\frac{\sum_{i=1}^{M_j^{\rm U}}{\gamma_{i,j}\kappa_{i,j}}}{C^{\rm bh}_{j,n}},
\end{equation}
where $\kappa_{i,j}$ and $\gamma_{i,j}$ are the computation task size in bit and the compression ratio of UE $Y_{i,j}$, respectively. Note that the computation task size is assumed to be fixed in this work, and the case of computation task size being random can be addressed by some extra work based on our derived results. On the other hand, when the compressed data arrives at the FAP, it will be firstly sent to the DD processor to retrieve the original data, and then sent to the task computing processor to obtain the computing results. If this process is finished before the new compressed data arrives at the FAP, this new data will follow the same process. Otherwise, the new arrived data has to wait in the queue to be processed in the DD and task computing processors. Therefore, the entire delay of the DD and task computing delay $T^{\rm dp}_{i,j,n}$ consists of the queueing delay and processing delay in the DD and task computing processors. We assume that the service times of DD and task computing both follow the exponential distribution \cite{ren2019data}. For the DD process, if decompressing one bit data of the computation task from UE $Y_{i,j}$ needs $g^{\rm d}_{i,j}$ central processing unit (CPU) cycles and the processor of the FAP $X_n^{\rm A}$ for DD can provide $s_n^{\rm {A,d}}$ CPU cycles per second, the service rate $\mu^{\rm dd}_n$ for DD in the FAP can be calculated by $\frac{\gamma_{i,j}\kappa_{i,j}g^{\rm d}_{i,j}}{s_n^{\rm {A,d}}}$. Similarly, by denoting the required CPU cycles for computing one bit data of the computation task from UE $Y_{i,j}$ as $g^{\rm p}_{i,j}$ and the CPU cycles that the processor of the FAP can provide for computing this task in a second as $s_n^{\rm {A,p}}$, the service rate $\mu^{\rm cp}_n$ for task computing in the FAP $X_n^{\rm A}$ can be calculated by $\frac{\kappa_{i,j}g^{\rm p}_{i,j}}{s_n^{\rm {A,p}}}$. Due to there exists significant differences in the uplink transmission delay and the DC processing delay for the three DC modes, the details of these two delays for each DC mode are discussed as follows.

\textbf{Local Compression Mode:} In this mode, the computation tasks are compressed at the UE before transmitting to the associated FN via the wireless link. The DC process can be modelled following the M/M/1 queue \cite{ren2019data}. By denoting the required CPU cycles for compressing one bit data and the provided CPU cycles per second of the processor in the UE $Y_{i,j}$ as $g^{\rm u,c}_{i,j}$ and $s^{\rm u,c}_{i,j}$, respectively, the service rate $\mu_{i,j}^{u,c}$ of the DC in this UE can be denoted by $\frac{\kappa_{i,j}g^{\rm u,c}_{i,j}}{s^{\rm u,c}_{i,j}}$. The DC delay that consists of queueing and processing delay in the UE $Y_{i,j}$ is denoted by $T_{i,j}^{\rm l,dc}$. In addition, for analytical tractability, the uplink transmission delay $T_{i,j}^{\rm l,ut}$ is evaluated by the average uplink transmission rate $C_{i,j}^{\rm ut}$ \cite{hu2019density} as follows:
\begin{equation}
T_{i,j}^{\rm l,ut}=\frac{\gamma_{i,j}\kappa_{i,j}}{C_{i,j}^{\rm ut}}.
\end{equation}

\textbf{Edge Compression Mode:} In this mode, the computation tasks that generated at the UE are firstly transmitted to the FN and then these computation tasks are compressed at the FN before transmitting to the FAP. For the computation tasks generated at UE $Y_{i,j}$, the uplink transmission delay $T_{i,j}^{\rm e,ut}$ for the edge compression mode can be calculated by
\begin{equation}
T_{i,j}^{\rm e,ut}=\frac{\kappa_{i,j}}{C_{i,j}^{\rm ut}}.
\end{equation}
Similar as the local compression mode, the DC process in the FN is modelled following the M/M/1 queue, and we denote the required CPU cycles for compressing one bit data and the provided CPU cycles per second of the processor in the FN $X_{j}^{\rm N}$ as $g^{\rm n}_{j}$ and $s^{\rm n,c}_{j}$, respectively. Accordingly, the service rate of DC $\mu_{i,j}^{n,c}$ in the FN $X_{j}^{\rm N}$ can be denoted by $\frac{\kappa_{i,j}g^{\rm u,c}_{i,j}}{s^{\rm n,c}_{j}}$. Moreover, the DC delay in the FN $X_j^{N}$ for the edge processing mode is denoted by $T_{i,j}^{\rm e,dc}$.

\textbf{Hybrid Compression Mode:} As illustrated in Fig. \ref{fig_cm_model}, the computation tasks generated at the UE $Y_{i,j}$ has a probability of $\beta_{i,j}$, namely COR, being compressed at the associated FN while the rest of computation tasks are uploaded to the FN for DC. After this, they are all sent to a pre-processing unit in the FN for transmitting on the constrained backhaul, wherever the computation tasks are compressed. To simplify the analysis, we assume that the computing capability of the pre-processing unit is relatively large as compared with other processing steps, thus the delay of the pre-processing is ignored in this work. Equipped with these assumptions, the average uplink transmission delay $T_{i,j}^{\rm h,ut}$ of UE $Y_{i,j}$ can be calculated by:
\begin{equation}
T_{i,j}^{\rm h,ut}=\beta_{i,j}T_{i,j}^{\rm e,ut}+(1-\beta_{i,j})T_{i,j}^{\rm l,ut}.
\end{equation}
The average DC delay $T_{i,j}^{\rm h,dc}$ for the hybrid compression mode can be obtained as follows:
\begin{equation}
T_{i,j}^{\rm h,dc}=\beta_{i,j}T_{i,j}^{\rm e,dc}+(1-\beta_{i,j})T_{i,j}^{\rm l,dc}.
\end{equation}

Therefore, the entire end-to-end delay $T_{i,j,n}^{\rm tot,\xi}$ for UE $Y_{i,j}$ to the FAP $X_{n}^{\rm A}$ in the F-RAN with DC for mode $\xi$ can be evaluated as:
\begin{equation}
T_{i,j,n}^{\rm tot,\xi}=T_{i,j}^{\rm \xi,ut}+T_{i,j}^{\rm \xi,dc}+T^{\rm bh}_{i,j}+T^{\rm dp}_{i,j,n},
\end{equation}
where $\xi\in\{\rm l,e,h\}$ respectively denotes the local compression, the edge compression and hybrid compression modes.


\subsection{Investigated Performance Metrics}
In this work, we mainly investigate the SDCP in the large-scale spatially-random deployed clustered F-RAN. The SDCP is defined based on two folds described as follows. One is the successful transmission probability (STP) $\Delta_{i,j}$ that the SIR of the uplink wireless link from the UE $Y_{i,j}$ to the FN $X_{j}^{\rm N}$ is larger than an SIR threshold $\tau$, i.e., $\Delta_{i,j}=\mathbb{P}({{\rm SIR}_{i,j}}>\tau)$, where ${\rm SIR}_{i,j}$ is defined as the uplink SIR of the UE $Y_{i,j}$ associated with the FN $X_{j}^{N}$. This ensures that the computation tasks can be uploaded to the FN. The other is the successful task execution probability (STEP) $\Xi_{i,j,n}^{\xi}$ that the entire end-to-end delay $T_{i,j,n}^{{\rm tot},\xi}$ defined in Section \ref{sec_sys_mod} is smaller than a target delay $\varrho$, i.e., $\Xi_{i,j,n}^{\xi}=\mathbb{P}_(T_{i,j,n}^{\rm tot,\xi}<\varrho)$. Equipped with these, the SDCP $\Theta_{i,j,n}^{\xi}$ of the UE $Y_{i,j}$ in the clustered F-RAN can be obtained as:
\begin{equation}
\label{equ_sdcp_def}
\Theta_{i,j,n}^{\xi}=\Delta_{i,j}  \Xi_{i,j,n}^{\xi}.
\end{equation}
The derivation of the STP and the STEP will be provided in the next section. The used notations and variables in this work are summarised in Table \ref{notations}.
\begin{table}
	\centering
	\caption{NOTATIONS AND VARIABLES}\label{notations}
	\begin{tabular}{|c|l|}
		\hline
		Notations & Definitions \\
		\hline
		$\lambda_{\rm N}$& The density of the FNs following GPP.\\
		\hline
		$M_j^{\rm U}$, $M_{\rm N}$ , $M_{\rm A}$ & The number of UEs, FNs and FAPs. \\
		\hline
		${C^{\rm bh}_{j,n}}$ & The backhaul capacity from the $j$-th FN to the $n$-th FAP. \\
		\hline
		$Y_{i,j}$ & The location of the $i$-th UE served by the $j$-th FN. \\
		\hline
		$X^{\rm N}_j$ & The location of the $j$-th FN. \\
		\hline
		$R_{i,j}$ & The distance between the UE $Y_{i,j}$ and the FN $X^{\rm N}_j$\\
		\hline
		$B$ & The subchannel bandwidth for each UE. \\
		\hline
		$T^{\rm bh}_{i,j}$ & The backhaul transmission delay of the $i$-th UE served by the $j$-th FN. \\
		\hline
		$T^{\rm dp}_{i,j,n}$ & The DD and task computing delay at the $n$-th FAP for the $i$-th UE served by the $j$-th FN. \\
		\hline
	    $T_{i,j}^{\rm l,ut}$ & The uplink transmission delay between the $i$-th UE and the $j$-th FN for local compression mode.  \\
		\hline
		$T_{i,j}^{\rm l,dc}$ & The local DC delay for the $i$-th UE served by the $j$-th FN. \\
		\hline
		$T_{i,j}^{\rm e,ut}$ & The uplink transmission delay between the $i$-th UE and the $j$-th FN for edge compression mode. \\
		\hline
		$T_{i,j}^{\rm e,dc}$ & The edge DC delay for the $i$-th UE served by the $j$-th FN. \\
		\hline
		$\beta_{i,j}$ & The compression offloading ratio for the $i$-th UE served by the $j$-th FN. \\
	    \hline
	    $g^{\rm u,c}_{i,j}$ & The required CPU cycles for compressing one bit data at the $i$-th UE served by the $j$-th FN. \\
	    \hline
	    $g^{\rm n}_{j}$ & The required CPU cycles for compressing one bit data at the $j$-th FN. \\
	    \hline
	    $g^{\rm d}_{i,j}$ & The required CPU cycles for decompressing one bit data from the $i$-th UE served by the $j$-th FN. \\
	    \hline
	    $g^{\rm p}_{i,j}$ & The required CPU cycles for computing one bit data from the $i$-th UE served by the $j$-th FN. \\
	    \hline
	     $\mu_{i,j}^{u,c}$ & The service rate of the DC process for the $i$-th UE served by the $j$-th FN. \\
	    \hline
	     $\mu_{i,j}^{n,c}$ & The service rate of the DC process for the $j$-th FN. \\
	    \hline
	     $\mu_n^{\rm dd}$ & The service rate of the DD process for the $n$-th FAP. \\
	    \hline
	      $\mu_n^{\rm cp}$ & The service rate of the cloud computing process for the $n$-th FAP \\
	    \hline
	    $\tau$ & The SIR threshold for the uplink transmission. \\
	    \hline
	    $s^{\rm u,c}_{i,j}$ & The basic frequency of the $i$-th UE served by the $j$-th FN.    \\
	    \hline
	    $s^{\rm n,c}_{j}$ & The basic frequency of the $j$-th FN.    \\
	    \hline
	    $s^{\rm A,d}_{n}$ & The basic frequency of the $n$-th FAP.    \\
	    \hline
	    $\gamma_{i,j}$    & The task packet compression ratio. \\
	    \hline
	    $p$ & The highest UE transmit power. \\
	    \hline
	    $\alpha$  & The pathloss exponent for the data transmission.\\
	    \hline
	\end{tabular}
\end{table}

\section{Performance Analysis}
\label{sec_perf_analysis}
In this section, firstly, we derive the uplink STP of the UE in the clustered F-RAN. Then the STEP is derived for each DC mode. 

\subsection{Successful Transmission Probability}
The STP indicates the probability of UE being capable of transmitting the computation tasks to the FN. For example, if the uplink SIR of the UE is poor, the computation tasks generating at the UE can only be processed locally. 
The STP of an arbitrary UE in the F-RAN is given in Lemma \ref{lemma_stp} as follows.
\begin{MyLem}
\label{lemma_stp}
In the clustered F-RAN with FNs and their associated UEs deployed following MCP, the STP of an arbitrary UE $Y_{i,j}$ is given as:
\begin{equation}
\label{equ_stp_res}
\Delta_{i,j}\!=\int_0^1{\exp\left[-{\pi\lambda_{\rm N}}\int_0^{\infty}\left(1-\mathcal{H}(v,x)\right){\rm d} x \right]{\rm d} v},
\end{equation}
where $\mathcal{H}(v,x)$ is defined as:
\begin{equation}
\mathcal{H}(v,x)\triangleq\frac{1}{2\pi}\int_{0}^{1}\int_{0}^{2\pi}\frac{1}{1+ \tau(\frac{u^{\epsilon}v^{1-\epsilon}}{cx+cu-2\sqrt{cxu}\cos y})^{\frac{\alpha}{2}}}{\rm d}y{\rm d}u.
\end{equation}
\end{MyLem}
\begin{proof}
See Appendix \ref{app_stp}.
\end{proof}
With the results given in \eqref{equ_stp_res}, the STP of the UE with FNs positioned following the PPP can be analysed. Nevertheless, this theoretical result is still time-consuming due to there exists a four-dimensional integration. To reduce the computational complexity, and note that the distances between the interfering UEs and their associated FAPs have trivial effect on the aggregate interference generated by the interfering UEs, we accordingly give the approximated results for \eqref{equ_stp_res} in Proposition \ref{prop_stp_appx} as follows.

\begin{MyProp}
\label{prop_stp_appx}
In the F-RAN, the approximated STP $\widetilde{\Delta}_{i,j}$ that the uplink SIR being larger than a threshold $\tau$ can be obtained by
\begin{equation}
\label{eq_stp_appx}
\widetilde{\Delta}_{i,j}=\left\{
\begin{aligned}
&E_{\frac{\epsilon}{\epsilon-1}}(\zeta)+\zeta^{\frac{1}{\epsilon-1}}\Gamma\left(1+\frac{1}{1-\epsilon}\right),  \quad 0\leq\epsilon<1,\\
&\exp(-\zeta), \quad \epsilon=1.
\end{aligned}
 \right.
\end{equation}
where $\zeta=\frac{2\pi^2\lambda_{\rm N}\tau^{\frac{2}{\alpha}} v^{1-\epsilon}}{c\alpha(1+\epsilon)\sin(\frac{2\pi}{\alpha})}$, $E_{\frac{\epsilon}{\epsilon-1}}(\zeta)$ is the two-argument exponential integral function defined as $\int_1^{\infty}\frac{e^{-\zeta t}}{t^{\frac{\epsilon}{\epsilon-1}}}{\rm d}t$.
\end{MyProp}
\begin{proof}
See Appendix \ref{appd_prop_stp_appx}.
\end{proof}

The computational complexity has been significantly reduced by removing the two-dimensional integration for the actual interfering UE positions. 

\subsection{Successful Task Execution Probability}

The STEP is defined as the entire end-to-end delay $T_{i,j,n}^{\rm tot, \xi}$ being smaller than a target delay $\varrho$. This target delay can be treated as the latency requirement of UE on obtaining the results of computation tasks. As the entire end-to-end delay $T_{i,j,n}^{\rm tot, \xi}$ is composed of the uplink transmission delay, the DC delay, the backhaul transmission delay and the DD and task computation delay, to obtain the successful task execution probability, these delay distributions for each DC mode should be investigated. Note that the backhaul transmission delay and the DD and task computation delay are the same for each mode. Thus we firstly derive the results for these two delays. Recall that the backhaul transmission delay from the FN $X_j^{\rm N}$ to the FAP $X_n^{\rm A}$ for each mode has already been given in \eqref{eq_delay_bh}. Next we will derive the distribution of DD and task computing delay in the FAP, which is another main technical part of this work.   

By assuming the service rates of DD and task computing both follow the exponential distribution, the total service rate of the two-step computing process, i.e., DD and task computing, in the FAP can be modelled following the hypo-exponential distribution \cite{mor2013performance}. The probability density function (PDF) of the total service rate $\mu_{n}^{\rm A}$ can be denoted by:
\begin{equation}
\label{equ_pdf_serv_time_fap}
f_{\mu_{n}^{\rm A}}(x) = \frac{\mu_n^{\rm dd}\mu_n^{\rm cp}}{\mu_n^{\rm dd}-\mu_n^{\rm cp}}\left[\exp(-\mu_n^{\rm cp}x)-\exp(-\mu_n^{\rm dd}x) \right].
\end{equation}
Equipped with this, the distribution of the entire delay $T_{i,j,n}$ spent on the DD and task computing in the FAP is given in Lemma \ref{lemma_lt_dd_tc_time} as follows.
\begin{MyLem}
\label{lemma_lt_dd_tc_time}
The distribution of the entire delay $T_{i,j,n}$ spent on the DD and task computing in the FAP $X_n^{\rm A}$, denoted by $f_{T_{i,j,n}^{\rm dp}}(t)$, is given as:
\begin{equation}
\label{eq_dp_delay_distr}
f_{T_{i,j,n}^{\rm dp}}(t)\!=\!\frac{\mu_n^{\rm dd}\mu_n^{\rm cp}\left(1-\rho \right) }{\eta_n}e^{\frac{(\Lambda_n^{\rm A}-\widehat{\mu}_n^{A})}{2}t}\left(e^{\frac{\eta_n}{2}t}-e^{-\frac{\eta_n}{2}t} \right),
\end{equation}
where $\Lambda_n^{\rm A}$ is defined as the aggregate computation task arriving rate at the FAP $X_{n}^{\rm A}$, the term $\widehat{\mu}_n^{A}=\mu_n^{\rm dd}\!+\!\mu_n^{\rm cp}$, and $\eta_n=\sqrt{\left(\Lambda_n^{\rm A}+\widehat{\mu}_n^{A}\right)^2-4\mu_n^{\rm dd}\mu_n^{\rm cp}}$ for denotational simplicity.
\end{MyLem}
\begin{proof}
See Appendix \ref{appd_lemma_lt_dd_tc_time}.
\end{proof}
With the results in \eqref{eq_delay_bh} and \eqref{eq_dp_delay_distr}, the results or distributions of the backhaul transmission delay the DD and task computation delay are obtained, which are both in closed-form. Equipped with these, the STEP of the UE $Y_{i,j}$ is given in Theorem \ref{theo_step} as follows. 
\begin{MyTheo}
\label{theo_step}
In the clustered F-RAN, the STEP $\Xi_{i,j,n}^{\xi}$ of the UE $Y_{i,j}$ associated with $X_j^{\rm N}$ using the compression mode $\xi$, $\xi\in\{\rm l,e,h\}$, can be calculated by
\begin{equation}
\label{equ_clo_step}
\begin{aligned}
&\Xi_{i,j,n}^{\xi}=\frac{\mu_n^{\rm dd}\mu_n^{\rm cp}\!-\!\Lambda_n^{\rm A}}{\eta_n}
  \left\{
  \left(1+\frac{\beta_{i,j}^{\xi} \sigma_{i,j}^{\rm U} }
  {\widehat{\beta}_{i,j}^{\xi}\sigma_{j}^{\rm N}-\beta_{i,j}^{\xi} \sigma_{i,j}^{\rm U}}\right)
  \left[\mathcal{A}\left(\varsigma_n^{\rm A},{\widehat{\beta}_{i,j}^{\xi}},{\sigma_{i,j}^{\rm U}},\varrho'_{\xi}\right) 
  -
  \mathcal{A}\left(\varsigma_n^{\rm A'},{\widehat{\beta}_{i,j}^{\xi}},{\sigma_{i,j}^{\rm U}},\varrho'_{\xi}\right) 
  \right] \right.\\ 
&\left.  \!+\!\frac{1}
  {\frac{\widehat{\beta}_{i,j}^{\xi}\sigma_{j}^{\rm N}}{\beta_{i,j}^{\xi} \sigma_{i,j}^{\rm U}}\!-\!1}
  \left[
  \mathcal{A}\left(\varsigma_n^{\rm A'},{\beta_{i,j}^{\xi}},{\sigma_{j}^{\rm N}},\varrho'_{\xi}\right) 
  \!-\!
  \mathcal{A}\left(\varsigma_n^{\rm A},{\beta_{i,j}^{\xi}},{\sigma_{j}^{\rm N}},\varrho'_{\xi}\right) 
  \right]
  \!-\!\left(
  \frac{2}{\varsigma_n^{\rm A}}e^{\frac{\varsigma_n^{\rm A}}{2} \varrho'_{\xi}}
  \!-\!\frac{2}{\varsigma_n^{\rm A'}}e^{\frac{\varsigma_n^{\rm A'}} {2} \varrho'_{\xi}}
  \!-\!\frac{2}{\varsigma_n^{\rm A}}\!+\!\frac{2}{\varsigma_n^{\rm A'}}
  \right)
  \right\},
\end{aligned}
\end{equation}
where $\varsigma_n^{\rm A}=\Lambda_n^{\rm A}-\widehat{\mu}_n^{A}-\eta_n$ and $\varsigma_n^{\rm A'}=\Lambda_n^{\rm A}-\widehat{\mu}_n^{A}+\eta_n$. For denotational simplicity, we define the function $\mathcal{A}(\cdot)$ as:
\begin{equation}
\mathcal{A}\left({a},{b},{c},x\right)=
	\frac{2b}{2 c + b a}
	e^{-\frac{cx}{b}}	
	\left(e^{\frac{2 c + ba}{2b} x}-1  \right),
\end{equation}
where the variables $a$, $b$, and $c$ follow $a\in\{\varsigma_n^{\rm A},\varsigma_n^{\rm A'}\}$, $b\in \{\beta_{i,j}^{\xi},\widehat{\beta}_{i,j}^{\xi}=1-\beta_{i,j}^{\xi} \}$ and $c\in\{\sigma_{i,j}^{\rm U},\sigma_{j}^{\rm N}\}$. The terms $\sigma_{j}^{\rm N}=\mu_{j}^{\rm n,c}-\Lambda_{j}^{\rm N}$ and $\sigma_{i,j}^{\rm U}=\mu_{i,j}^{\rm u,c}-\Lambda_{i,j}$. Additionally, the term $\varrho'_{\xi}$ denotes the delay threshold for task compression, decompression and processing in the compression mode $\xi$, which can be expressed as $\varrho'_{\xi}=\varrho-\bar{T}_{i,j}^{\rm \xi,ut}-T^{\rm bh}_{j,n}$. The term $\bar{T}_{i,j}^{\rm \xi,ut}$ can be calculated by:
\begin{equation}
\label{equ_avg_up_trans}
\bar{T}_{i,j}^{\rm \xi,ut}=\frac{\beta_{i,j}\gamma_{i,j}\kappa_{i,j}+\widehat{\beta}_{i,j}\kappa_{i,j}}{\int_0^{\infty}\log_2(1+\tau)\widetilde{\Delta}_{i,j}(\tau){\rm d}\tau},
\end{equation}
and the term $T^{\rm bh}_{j,n}$ is given in \eqref{eq_delay_bh}.
\end{MyTheo}
\begin{proof}
See Appendix \ref{appd_theo_sdcp}.
\end{proof}
Equipped with the results given in Proposition \ref{prop_stp_appx} and Theorem \ref{theo_step}, the SDCP $\Theta_{i,j,n}$ of the clustered F-RAN is given in Corollary \ref{coro_sdcp} as follows:
\begin{MyCoro}
\label{coro_sdcp}
In the clustered F-RAN, the SDCP $\Theta_{i,j,n}^{\xi}$ of the UE $Y_{i,j}$ associated with $X_j^{\rm N}$ using the compression mode $\xi$, $\xi\in\{\rm l,e,h\}$, can be calculated by
\begin{equation}
\label{Coro 1}
\Theta_{i,j,n}^{\xi}=\left\{
\begin{aligned}
&\left[E_{\frac{\epsilon}{\epsilon-1}}(\zeta)+\zeta^{\frac{1}{\epsilon-1}}\Gamma\left(1+\frac{1}{1-\epsilon}\right)\right]\Xi_{i,j,n}^{\xi},  \quad 0\leq\epsilon<1,\\
&\exp(-\zeta)\Xi_{i,j,n}^{\xi}, \quad \epsilon=1.
\end{aligned}
 \right.
\end{equation}
\end{MyCoro}
\begin{proof}
Recall that the definition of SDCP is given in \eqref{equ_sdcp_def}. By incorporating the results given in Proposition \ref{prop_stp_appx} and Theorem \ref{theo_step}, the SDCP can be yielded. 
\end{proof}

\subsection{Optimal Compression Offloading Ratio}
For the hybrid compression mode, the proportion $\beta_{i,j}$ of task compressed at the FN may have significant effects on the SDCP, especially when the computing capability of the FN is limited. With the SDCP of hybrid mode given in Corollary \ref{coro_sdcp}, the optimal compression proportion at the FN $\beta_{i,j}^{\rm opt}$ to maximize the SDCP in the clustered F-RAN can be expressed as:
\begin{align}
&\arg\max_{\beta_{i,j}}\Theta_{i,j,n}^{\rm h}\\
{\rm s.t.}\quad &0\leq \beta_{i,j} \leq 1.
\end{align}
Note that the STP is not related to the $\beta_{i,j}$, thus the original optimization problem can be transformed as:
\begin{align}
&\arg\max_{\beta_{i,j}}\Xi_{i,j,n}^{\rm h}\\
{\rm s.t.}\quad &0\leq \beta_{i,j} \leq 1.
\end{align}
To simplify the analysis in the large-scale F-RAN, we assume that the compression offloading ratio of each UE is the same \cite{ren2019data}. Due to the complexity of $\Xi_{i,j,h}$, it is difficult to obtain the closed-form optimal compression proportion $\beta_{i,j}^{\rm opt}$. Fortunately, as the SDCP has been obtained in a closed-form and the $\beta_{i,j}$ has a limited range, the effect of compression ratio $\beta_{i,j}^{\rm opt}$ on the SDCP can be analysed numerically, and the optimal compression ratio can be obtained by the bisection method.

\section{Simulation Results}
\label{sec_simu_res}
In this section, firstly, the derived theoretical SDCP is validated by the Monte Carlo simulations. Based on the derived SDCP, the effects of FN computing capability, the COR and different compression mode in terms of the SDCP are analysed numerically. 

\subsection{The SDCP Validation}

\begin{table}
	\centering
	\caption{SIMULATION PARAMETERS}\label{table_simu_val}
	\begin{tabular}{|c|c|c|c|}
		\hline
		Parameters&Values & Parameters & Values\\
		\hline
		$c$&$10^{-4}$ FNs/m$^2$ &$\alpha$ & $4$\\
		\hline
		$B$ & $5$ MHz & $M_j^{\rm U}$ & $4$\\
		\hline
		$M_{\rm A}$ & $1$ & $M_{\rm N}$ & $2$  \\
		\hline
		$\kappa_{i,j}$ & $2048$ bits & $\gamma_{i,j}$ & 0.6 \\
		\hline
		$\epsilon$ & $0.8$ & $\tau$ & 0 dB\\
		\hline
		$s^{\rm u,c}_{i,j}$, $s^{\rm n,c}_{j}$ & $1$ GHz, $5$ GHz & $s_n^{\rm {A,d}}$ & $24$ GHz\\
		\hline
		 	$g^{\rm u, c}_{i,j}$, $g^{\rm n}_{j}$, $g^{\rm d}_{i,j}$ & $1 \times 10^{5}$ cycles & $g^{\rm p}_{i,j}$ & $1.5 \times 10^{5}$ cycles \\
		\hline
	\end{tabular}
\end{table}
The Monte Carlo simulations are applied to validate the theoretical SDCP. The theoretical SDCP is obtained following \eqref{Coro 1}. Due to the independence of the STP and STEP, the SDCP is validated by two Monte Carlo simulations. For the STP, in each iteration of Monte Carlo simulations, we place the target FN at the origin and place other FNs following the PPP. Each FN exists an interfering UE located around its associated FN in a circle region uniformly with radius being $1/\sqrt{c}$. By calculating the SIR of target FAP in each iteration, the cumulative distribution function (CDF) of the SIR for the target FAP can be obtained with $10,000$ iterations. As a result, the simulation result of STP for the target FAP can be obtained. For the STEP, $10,000$ arrivals are modelled following the Poisson process and the service time of each arrival in each processor is modelled following independent exponential distribution, whose parameter is related to the  what kind of task being processed and the computing capability of the processor. For each arrival, the arriving time at the buffer in the UE and the leaving time at the last processor in the FAP are recorded. Thus, the CDF of the task respond time, i.e., STEP, in the system can be obtained. Consequently, the simulation results of the SDCP can be yielded. Note that the M/G/1 and M/M/1 queues needs to be stable, thus the simulation parameters should satisfy $\mu_{i,j}^{u,c}\geq \Lambda_{i,j}$, $\mu_{i,j}^{n,c}\geq M_j^{\rm U} \Lambda_{i,j}$, $\mu_n^{\rm dd} \geq  M_{\rm N} M_j^{\rm U} \Lambda_{i,j}$, and $\mu_n^{\rm cp} \geq  M_{\rm N} M_j^{\rm U} \Lambda_{i,j}$. To simplify the analysis, it is assumed that the task generation rate and compression offloading ratio in each UE are the same, respectively. Moreover, all UEs, FNs and FAPs respectively have the same computing capability \cite{ren2019data}. 
The value of simulation parameters are listed in Table \ref{table_simu_val} unless otherwise specified. 
\begin{figure}
	\centering
	\includegraphics[width = 0.56\linewidth]{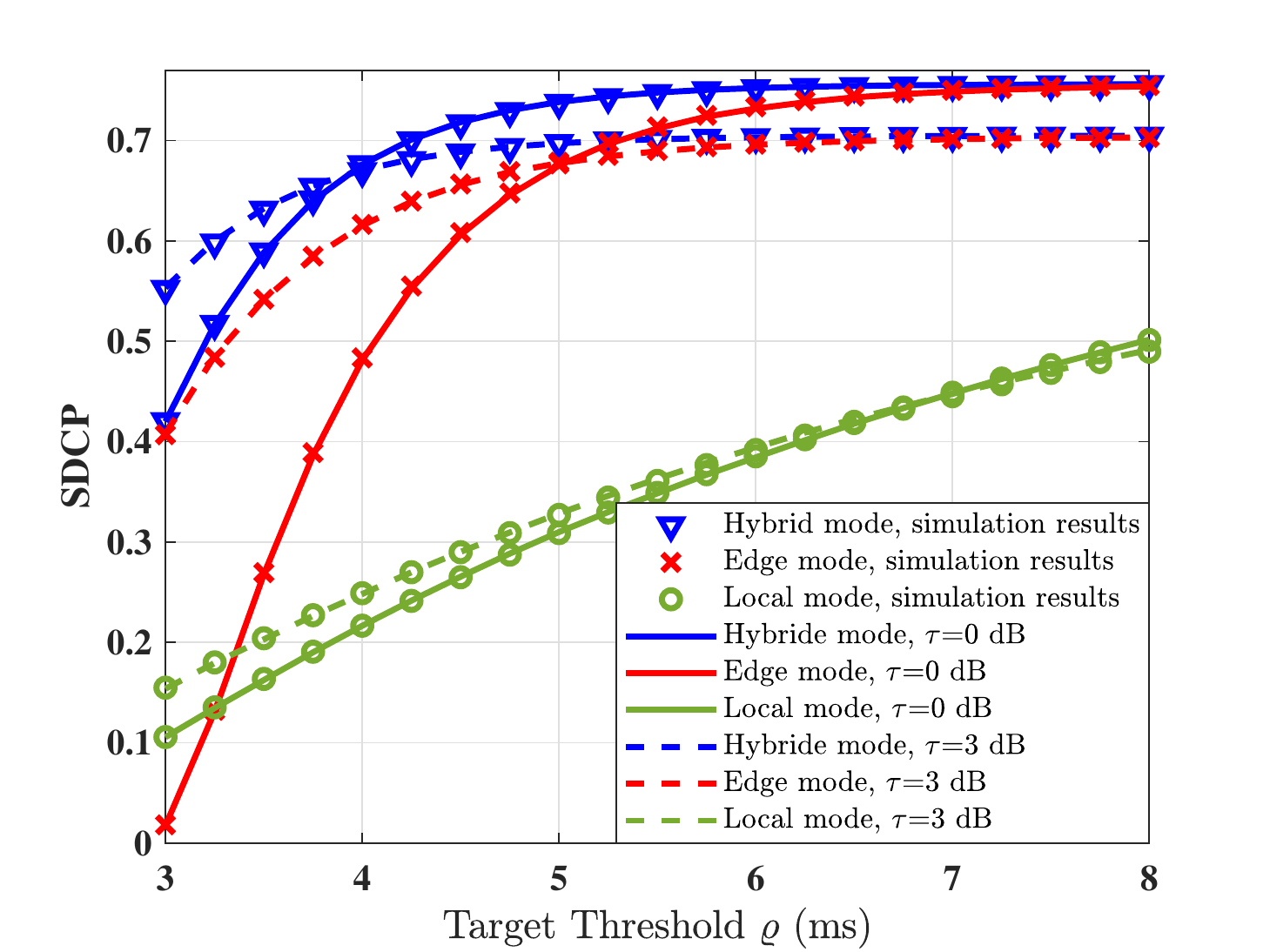}
	\caption{The SDCP validation for the local, edge and hybrid compression modes versus target latency with different SIR thresholds}
		\label{fig_validation_1}	
\end{figure}

In Fig. \ref{fig_validation_1}, the theoretical and simulation results of SDCP are both plotted versus the target latency with the SIR threshold $\tau$ being $0$ and $3$ under local, edge and hybrid compression modes. It is shown that the theoretical results match with the simulation results, which verifies the correctness of our derived results. In addition, the results indicate that in high latency requirement scenarios, the SDCP with a higher SIR threshold outperforms that with a lower SIR threshold and vice versa. This is mainly because that the uplink transmission delay has significant effect on the SDCP when the target latency is low, and the uplink transmission rate can be improved with a higher SIR threshold. In other words, for the F-RAN requiring a low target latency, the SIR threshold $\tau$ should be increased to improve the uplink transmission rate, which leads to the SDCP performance enhancement.

\begin{figure}
	\centering
	\includegraphics[width = 0.56\linewidth]{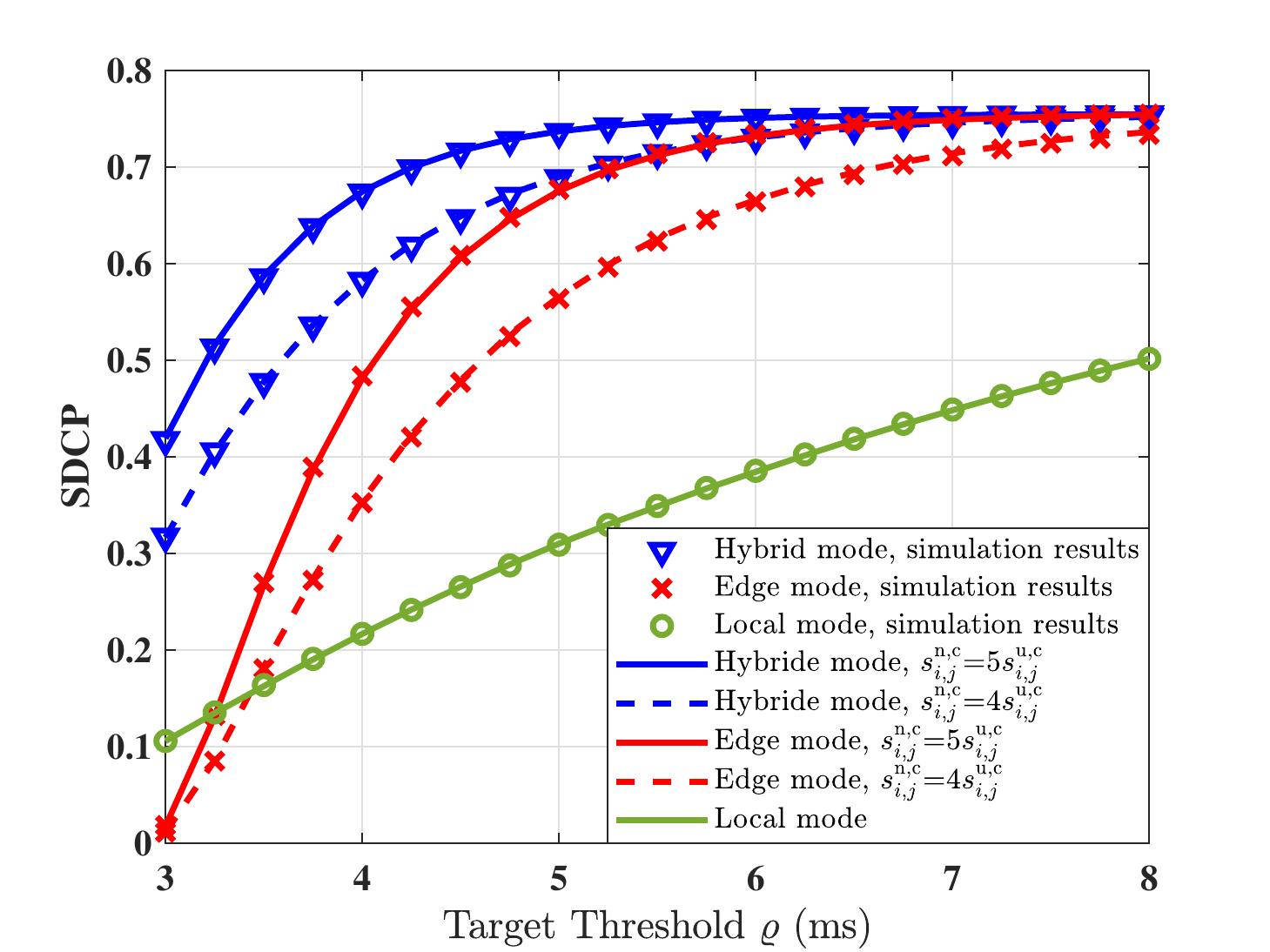}
	\caption{The SDCP validation for the local, edge and hybrid compression modes versus target latency with different computation capacities  }	
	\label{baisc_frequency_compare}
\end{figure}

In Fig. \ref{baisc_frequency_compare}, the theoretical and simulation results of SDCP versus the target latency with different FN computing capabilities. It is found that the computing capability of FN has greater effect on the edge compression mode as compared with the hybrid compression mode in terms of the SDCP. This is because that all the computation tasks are compressed at the FNs under the edge compression mode while only a partial of computation tasks are compressed at the FNs under the hybrid compression mode. Moreover, the results further verify the correctness of our derived theoretical SDCP, which enables us perform following analysis based on our derived theoretical results. 



\subsection{Numerical Analysis}
\begin{figure}
	\centering
	\includegraphics[width = 0.56\linewidth]{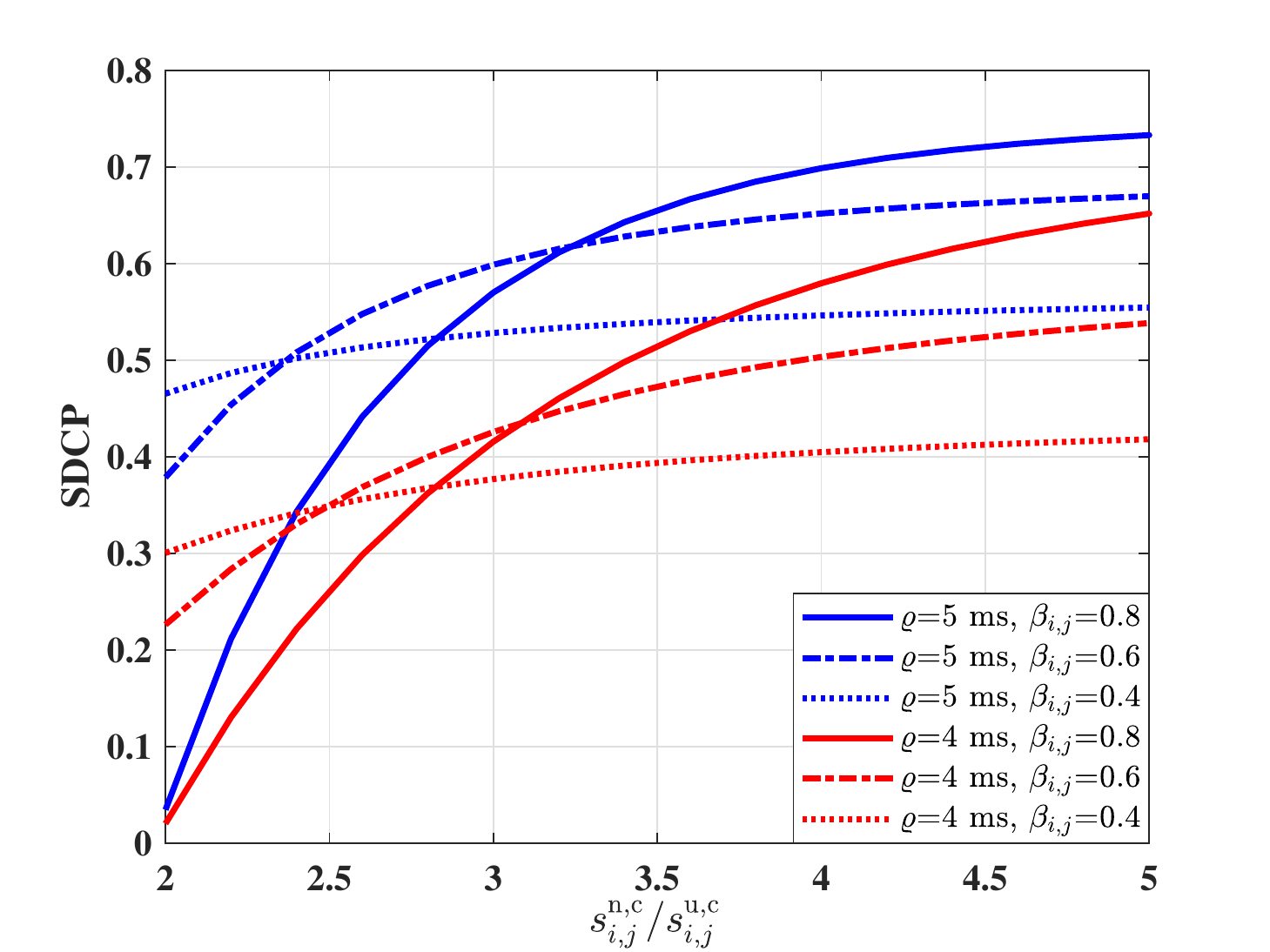}
	\caption{The effect of computing capability ratio between the FN and UE on the SDCP for the hybrid compression mode under compression ratio being $0.4$, $0.6$ and $0.8$ and target latency being $3$ and $5$ ms}	
	\label{prop_UE_FN}
\end{figure}
In Fig. \ref{prop_UE_FN}, the SDCPs are illustrated versus the computing capability ratio between the FN and UE, i.e., $\frac{s^{\rm n,c}_{i,j}}{s^{\rm u,c}_{i,j}}$, for the hybrid compression mode with compression ratio being $0.4$, $0.6$ and $0.8$ and target latency being $4$ and $5$ ms. It is shown that the SDCP monotonically increases with the computing capability enhancement in the FN. Moreover, the compression offloading ratio has a significant effect on the SDCP. Under the same simulation environment, the maximum gap between the SDCPs with different compression offloading ratios is $0.26$ when the target latency is $4$ ms. Furthermore, even with a strong computing capability in the FN, the SDCP enhancement is limited with a low compression offloading ratio. This is due to that the computing capability in the FN cannot be fully utilized when only a limited number of tasks being offloaded to the FN for compression. 
\begin{figure}
	\centering
	\includegraphics[width = 0.56\linewidth]{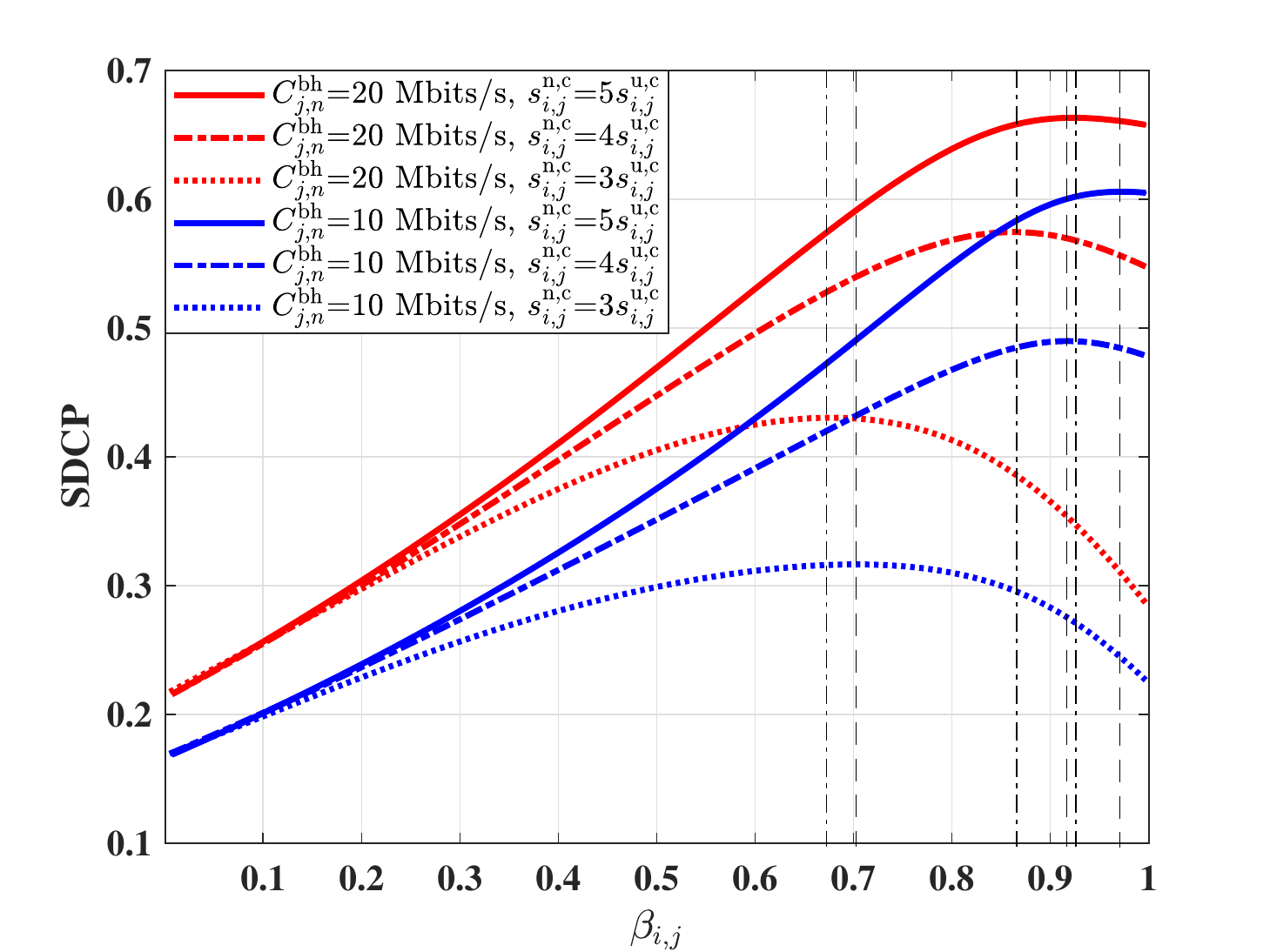}
	\caption{The SDCP versus compression offloading ratio under the hybrid compression mode with different backhaul capabilities and FN computing capabilities}
	\label{beta_analysis}
\end{figure}

In Fig. \ref{beta_analysis}, the SDCPs are plotted versus the compression offloading ratio with different backhaul capacities and FN computing capabilities under the hybrid compression mode. The results show that there exists an optimal compression offloading ratio to maximize the SDCP. In addition, with the same FN computing capability, the optimal compression offloading ratio decreases with the increasing of backhaul capacity. This is mainly because that the transmission latency in the backhaul is reduced by enhancing the backhaul capacity. As a result, more computation tasks can be compressed at the UE, which increases the processing latency in UE but reduce the uplink wireless transmission latency, thus increases the SDCP in the F-RAN. Moreover, as compared with the local and edge compression modes, i.e., $\beta_{i,j}=0$ and $\beta_{i,j}=1$, the SDCP can be enhanced with a maximum value of $0.43$ and $0.15$ under the hybrid compression mode, respectively. This indicates that the compression offloading ratio significantly affects the SDCP in the F-RAN. Next, the optimal compression offloading ratio to maximize the SDCP will be analysed. 


\subsection{The Optimal SDCP}

\begin{figure}
	\centering
	\includegraphics[width = 0.56\linewidth]{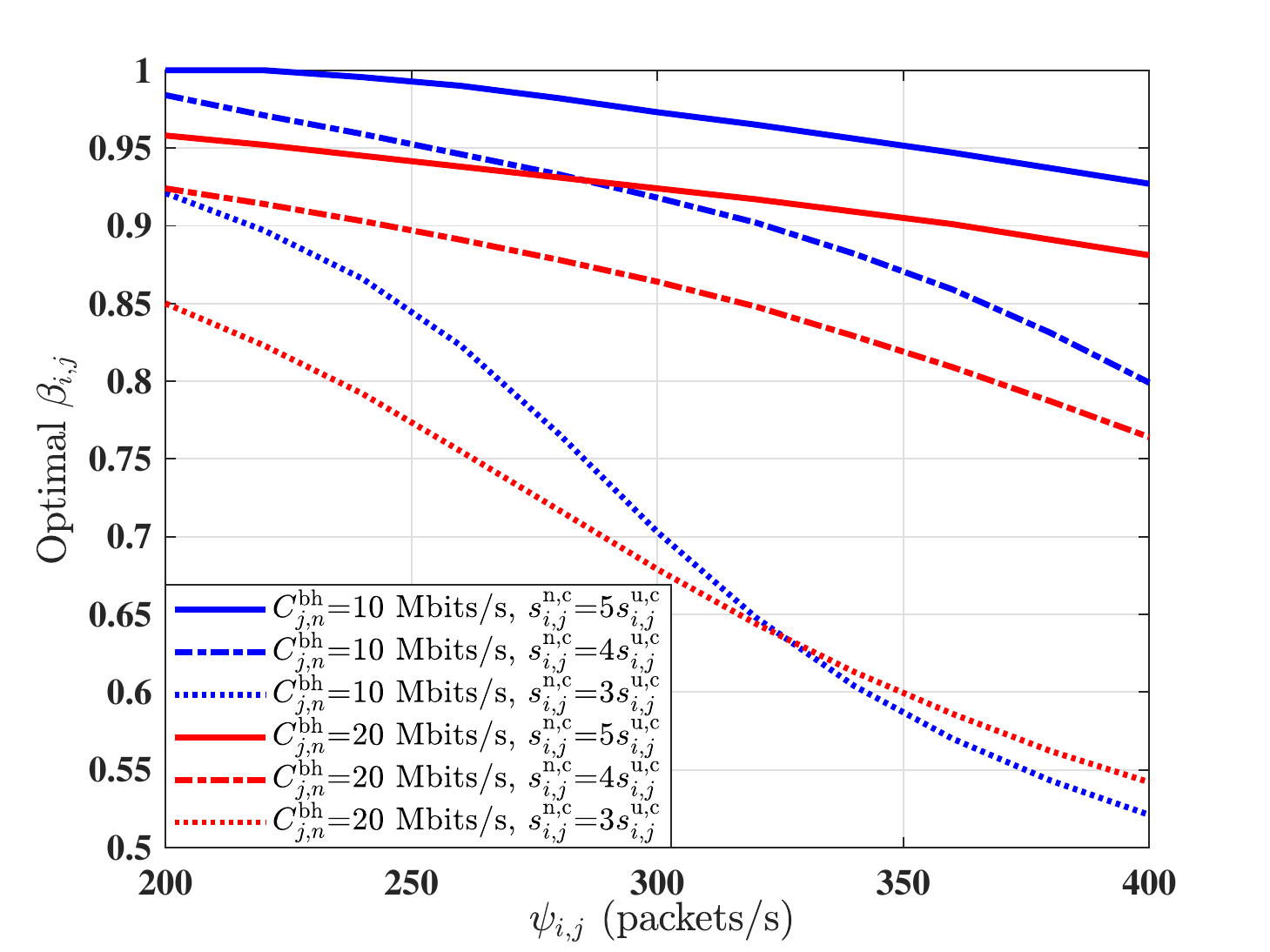}
	\caption{The relationship between the optimal $\beta_{i,j}$ and the task generation rate $\psi_{i,j}$ for the hybrid compression mode. The values of $M_j^{\rm U}$ are set to 3 and 5 for the red lines and blue lines, respectively. The SINR threshold $\tau$ is set to 0 dB, and the target latency is 4 ms. }	
	\label{fig_optimal_beta}
\end{figure}

In Fig. \ref{fig_optimal_beta}, the optimal compression offloading ratio is analysed versus the task generation rate with different backhaul capacities and FN computing capabilities. The results show that the optimal compression offloading ratio is monotonically decreased with the increasing of task generation rate. This is mainly because that the computing burden of FNs increases with the increasing of task generation rate, and offloading some computing tasks to be compressed at the UE can reduce the compression latency. Additionally, the optimal compression offloading ratio ranges from $0.5$ to $1$ under our simulation environment. In most cases, the optimal compression offloading ratio with low backhaul capability is higher than that with high backhaul capacity. Interestingly, the exception occurs in the scenarios with relatively low FN computing capability and task generation rate exceeding $325$ packets per second. This is because that under such a scenario, the main constraints on the SDCP becomes the DC compression latency in the UE and FN, thus by offloading more computation tasks compressed at the edge can increase the SDCP.

\begin{figure}
	\centering
	\label{uni}
	\subfigure[]{
		\includegraphics[width = 0.42\linewidth]{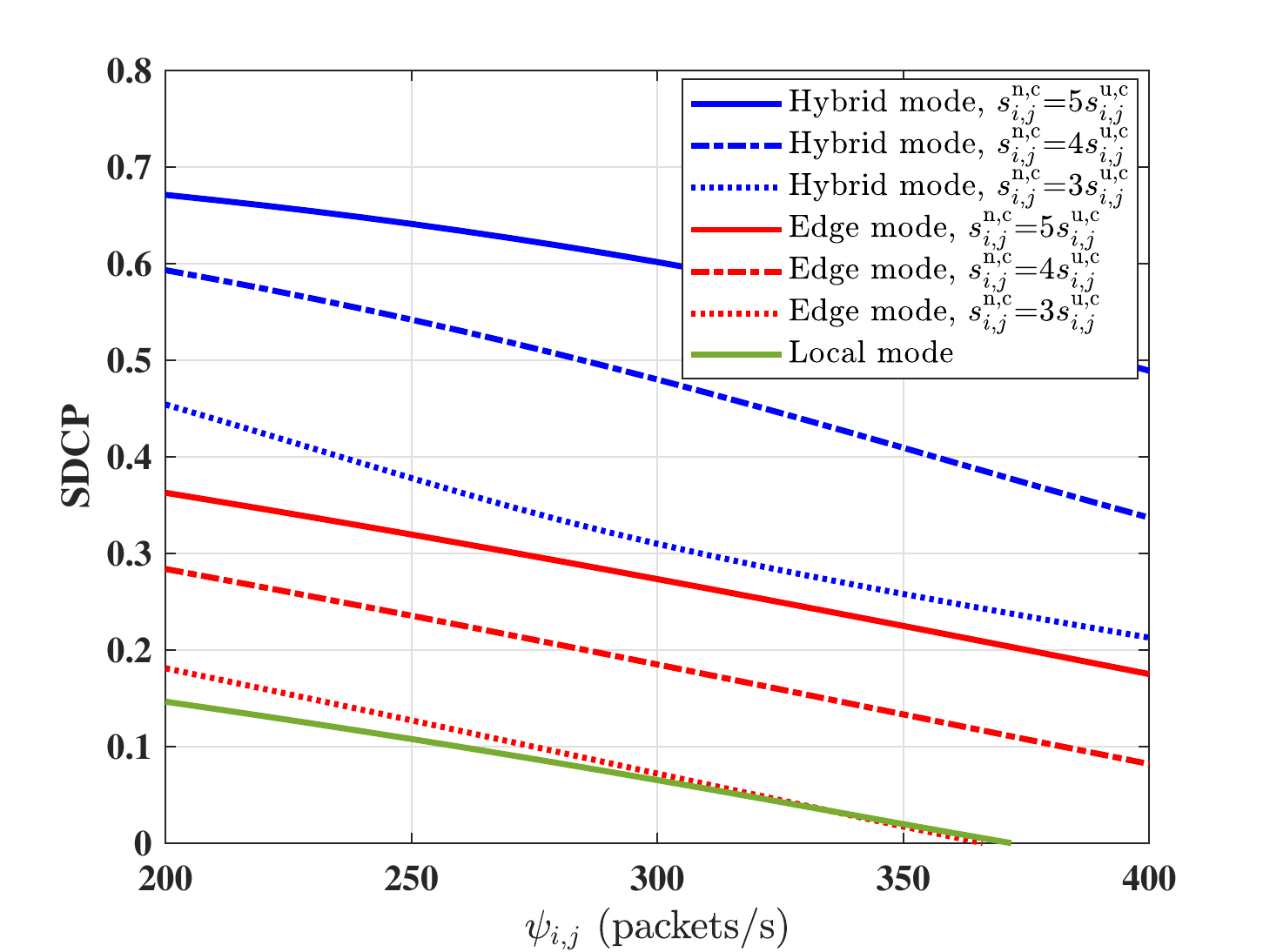}
	}	
	\subfigure[]{
		\includegraphics[width = 0.42\linewidth]{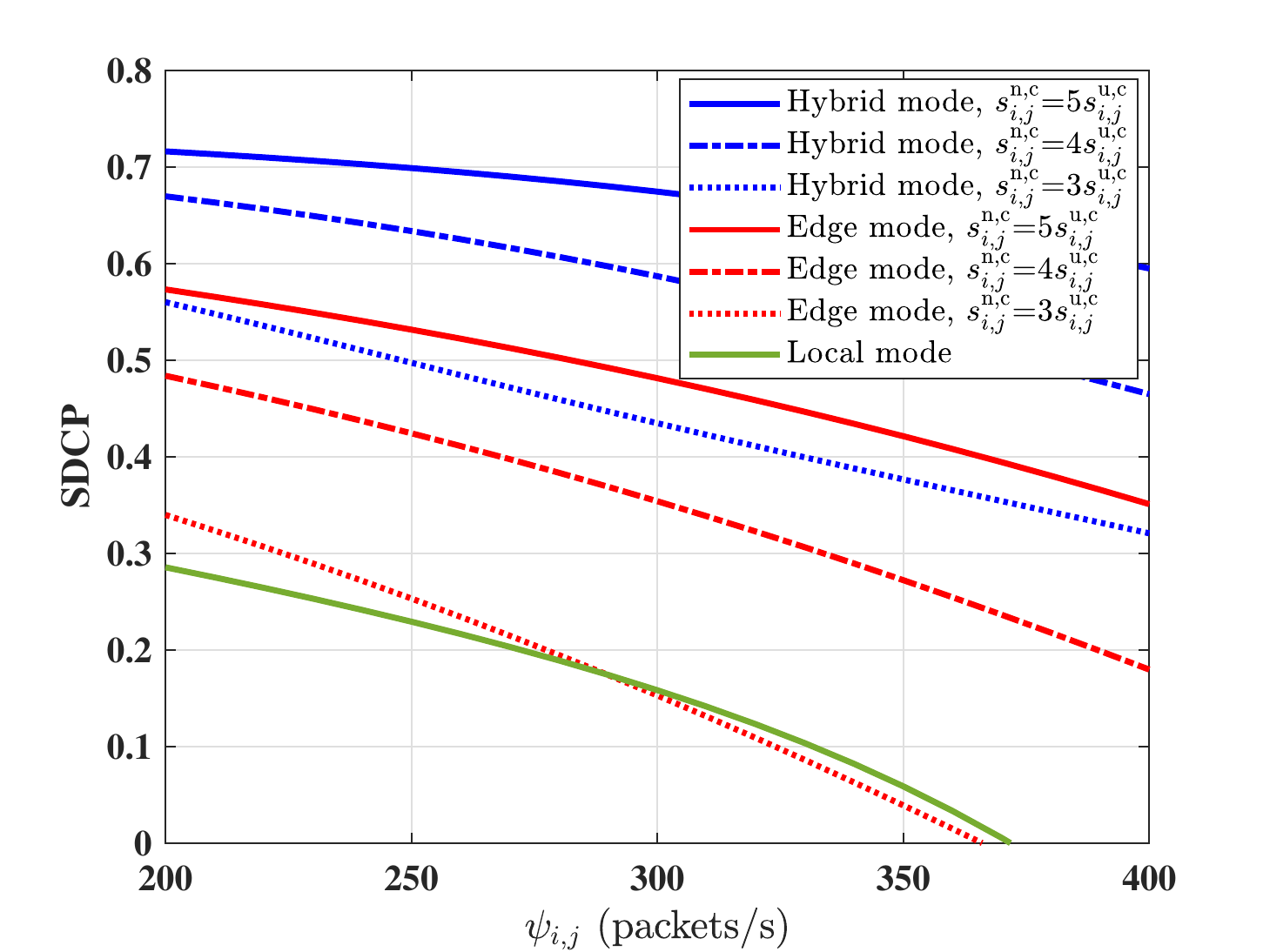}
	}
	\caption{The theoretical results of SDCP with task generation rate $\psi_{i,j}$ from $200$ to $400$ packets/s. The SINR threshold $\tau$ is set to 0 dB, and the target latency is 4 ms. The values of $C^{\rm bh}_{j,n}$ are set to $10$ Mbits/s and $20$ Mbits/s for (a) and (b), respectively.	} 
	\label{arrive_rate}
\end{figure}

Equipped with the optimal compression offloading ratio, in Fig. \ref{arrive_rate}, the SDCPs are illustrated versus the task generation rate under the hybrid compression mode with different FN computation capabilities under the backhaul capacity being $10$ and $20$ Mbits/s, respectively. The SDCP under the local and edge compression modes with corresponding simulation parameters are also plotted in these two figures.  
The results show that the hybrid compression mode outperforms the local and edge compression modes in terms of the SDCP. The SDCP can be improved by a maximum value of $0.55$ and $0.3$ in the hybrid compression mode as compared with the local and edge compression modes, respectively. It indicates that the hybrid compression mode, which can fully utilize the computing capability in UEs and FNs, should be adopted to enhance the SDCP of the F-RAN.
Additionally, on one hand, by increasing the backhaul capacity from $10$ to $20$ Mbits/s, the SDCPs obtained from the local and hybrid compression modes can be enhanced by a maximum value of $0.13$ and $0.22$, respectively. On the other hand, the SDCP obtained from the hybrid compression mode can be maximumly increased by $0.12$ by increasing the backhaul capacity from $10$ to $20$ Mbits/s. This indicates that by employed the hybrid compression mode, the requirement on the backhaul capacity can be reduced as compared with the local and hybrid compression mode.

\section{Conclusions}
\label{sec_conclusion}
In this work, we have obtained the closed-form results of the SDCP in the clustered DC-enabled F-RAN. Based on the results, the effect of FN computing capability, COR, and backhaul capacity on the SDCP are analysed numerically. Specifically, by given the feasible range of COR, the optimal COR for maximizing the SDCP can be obtained via the bisection method. It is observed that with the optimal COR, the SDCP under the hybrid compression mode can be improved with a maximum value of  $0.55$ and $0.3$ as compared with that under the local and edge compression mode, respectively, and the optimal COR lies in the range of $[0.55,1]$ under our simulation environment. Thus, the hybrid compression mode should be applied in the clustered F-RAN to enhance the SDCP. Additionally, the optimal COR decreases with the increasing of backhaul capacity. Furthermore, for the system requiring the minimal latency, the pressure of backhaul capacity can be alleviated by utilizing the hybrid compression mode. As the relationship between the compression ratio and the required CPU cycles for compressing one bit data is still unknown, this relationship and the effect of compression ratio on the latency performance can be analysed in our future work. Moreover, the energy efficiency and computing efficiency of the DC-enabled F-RAN are another two interesting topics to be analysed in the future.

\appendices

\section{Proof of Lemma \ref{lemma_stp}}
\label{app_stp}
Recall that we can place the target FN at the origin according to the Slyvnyak's theorem \cite{andrews2011tractable}. Additionally, the UEs are scattered uniformly in a circle region around their associated FNs and the group of UEs associated with the same FN are allocated with orthogonal frequency resources. Consequently, the STP $\Delta_{i,j}$ of an arbitrary UE can be obtained by calculating the following equation as:
\begin{equation}
\Delta_{i,j}=\mathbb{P}({\rm SIR}>\tau).
\end{equation}
Combining with \eqref{eq_sir_def}, the STP $\Delta_{i,j}$ can be transformed as:
\begin{equation}
\label{equ_app_a_1}
\begin{aligned}
\Delta_{i,j}&=\mathbb{P}\left[\frac{h_0Y_0^{-\alpha(1-\epsilon)}}{\sum_{k=1}^{\infty}h_kR_k^{\alpha\epsilon} Y_k^{-\alpha}}>\tau \right]\\
&=\mathbb{E}\left[\exp\left(-\frac{\mu\tau}{Y_0^{-\alpha(1-\epsilon)}}\sum_{k=1}^{\infty}h_kR_k^{\alpha\epsilon}Y_k^{-\alpha}\right) \right].\\
\end{aligned}
\end{equation}
Note that the square of distance $|Y_k|^2$ from the interfering UE to the target FN can be expressed as $|X_k^{\rm N}|^2+R_k^2-2|X_k^{\rm N}|R_j\cos(\gamma_k)$, where $\gamma_k$ is the angle between the link of interfering UE and its associated FN and the link of this associated FN and the target FN. Due to the FN positions are deployed following the PPP and the interfering UE distributed uniformly in a circle region around its associated FAP, 
the angle $\gamma_k$ follows a uniform distribution ranging from $0$ to $2\pi$, i.e., $\gamma_k\sim U(0,2\pi)$. In addition, as the UE distributes uniformly in a circle region around its associated FAP with radius $1/\sqrt{c}$, the square of distance between the UE and its associated FAP follows a uniform distribution in the range of $(0,1/c)$, i.e, $|Y_0|^2\sim U(0,1/c)$ and $|R_k|^2\sim U(0,1/c)$. Therefore, the result in \eqref{equ_app_a_1} can be derived as: 
\begin{equation}
\label{eq_appd_a_2}
\begin{aligned}
&\Delta_{i,j}=c\int_0^{\frac{1}{c}}\prod_{k=1}^{\infty}\left\{\mathbb{E}\left[\exp\left(-\frac{\mu\tau h_kR_k^{\alpha\epsilon}Y_k^{-\alpha}}{v^{-\frac{\alpha}{2}(1-\epsilon)}}\right) \right] \right\}{\rm d}v\\
&\overset{(a)}=\!c\int_0^{\frac{1}{c}}e^{\!-\!\pi\lambda_{\rm N}\int_0^{\infty}1\!-\!\mathbb{E}\left[ e^{-\mu\tau h_k\mathcal{D}(R_k,v,x,\gamma_k)^{\frac{\alpha}{2}}
 } \right]{\rm d}x }{\rm d}v\\
&= \int_0^{1}e^{\!-\!\pi\lambda_{\rm N}\int_0^{\infty}1\!-\!\mathbb{E}\left[ e^{-\mu\tau h_k\mathcal{D}(R_k,cv,x,\gamma_k)^{\frac{\alpha}{2}}
 } \right]{\rm d}x }{\rm d}v,
\end{aligned}
\end{equation}
where step (a) is obtained following the probability generating function (PGF) of the PPP \cite{andrews2011tractable}, and $\mathcal{D}(R_k,cv,x,\gamma_k)=\frac{R_k^{2\epsilon}v^{1-\epsilon}c^{1-\epsilon}}{x+R_k^2-2\sqrt{x}R_k\cos{\gamma_k} }$. By averaging on the independent random variables $h_k$, $R_k$ and $\gamma_k$, the expectation result in \eqref{eq_appd_a_2} can be calculated by:
\begin{equation}
\label{eq_appd_a_3}
\begin{aligned}
&\mathbb{E}\left[ e^{-\mu\tau h_n\mathcal{D}(R_n,v,x,\gamma_n)^{\frac{\alpha}{2}}
 } \right]\\
&=\frac{c}{2\pi}\!\int_0^{\frac{1}{c}}\!\int_0^{2\pi}\!\left[{1\!+\!\tau\left(\frac{c^{1-\epsilon}v^{1-\epsilon}u^{\epsilon}}{x\!+\!u\!-\!2\sqrt{xu}\cos{y}} \right)^{\frac{\alpha}{2}}}\right]^{\!-\!1}{\rm d}y{\rm d}u\\
&=\frac{1}{2\pi}\int_{0}^{1}\int_{0}^{2\pi}\frac{1}{1+ \tau(\frac{u^{\epsilon}v^{1-\epsilon}}{cx+cu-2\sqrt{cxu}\cos y})^{\frac{\alpha}{2}}}\text{d}y\text{d}u\\
&\triangleq\mathcal{H}(v,x).
\end{aligned}
\end{equation}  
By incorporating the result in \eqref{eq_appd_a_3} into \eqref{eq_appd_a_2}, and some symbolic transformations, the result in Lemma \ref{lemma_stp} can be achieved.


\section{Proof of Proposition \ref{prop_stp_appx}}
\label{appd_prop_stp_appx}
By treating the distances between the target FN and the neighbouring FNs as the distance between the target FN and the interfering UEs that used in the aggregate interference power, the approximated STP $\widetilde{\Delta}_{i,j}$ can be calculated by substituting $Y_k$ in \eqref{eq_appd_a_2} with $X_k^{\rm N}$ as:
\begin{equation}
\begin{aligned}
&\widetilde{\Delta}_{i,j}=c\int_0^{\frac{1}{c}}\!\prod_{k=1}^{\infty}\left\{\mathbb{E}\left[\exp\left(-\frac{\mu\tau h_kR_k^{\alpha\epsilon}|X_k^{\rm N}|^{-\alpha}}{v^{-\frac{\alpha}{2}(1-\epsilon)}}\right) \right] \right\}{\rm d}v\\
&=c\int_0^{\frac{1}{c}}\exp\left(-c\pi\lambda_{\rm N}\!\int_0^{\infty}\!\int_0^{\frac{1}{c}\!}{\frac{1}{1+\tau^{-1}\left(\frac{x}{v^{1-\epsilon}u^{\epsilon}} \right)^{\frac{\alpha}{2}}}}{\rm d}u{\rm d}x \right){\rm d}v\\
&=\int_0^{1}\exp\left(-\pi\lambda_{\rm N}\!\int_0^{1}\!\int_0^{\infty}\!{\frac{1}{1+\tau^{-1}\left(\frac{cx}{v^{1-\epsilon}u^{\epsilon}} \right)^{\frac{\alpha}{2}}}}{\rm d}x{\rm d}u \right){\rm d}v\\
&=\int_0^{1}\exp\left(-\frac{2\pi^2\lambda_{\rm N}\tau^{\frac{2}{\alpha}} v^{1-\epsilon}}{c\alpha\sin(\frac{2\pi}{\alpha})} \int_0^{1}u^{\epsilon}{\rm d}u \right){\rm d}v\\
&=\int_0^{1}\exp\left(-\zeta v^{1-\epsilon}\right){\rm d}v.
\end{aligned}
\end{equation}
By calculating this result in Wolfram Mathematica, the approximated result of $\widetilde{\Delta}_{i,j}$ can be obtained. 

\section{Proof of Lemma \ref{lemma_lt_dd_tc_time}}
\label{appd_lemma_lt_dd_tc_time}
Since the whole queueing model in the F-RAN meets the definition of Jackson network \cite{mor2013performance}, the aggregate arriving rate of the computation task at the FAP follows the Poisson process. The aggregate arriving rate of the computation tasks at the FAP $X_n^{\rm A}$ is denoted by $\Lambda_n^{\rm A}$. Thus, the whole processing in the FAP can be treated as the M/G/1 queueing model. According to the Pollaczek-Khinchin transform equation for the M/G/1 queueing model \cite{mor2013performance}, the Laplace transform of the entire delay on the DD and task computing can be represented as:
\begin{equation}
\label{eq_dd_cp_lap_def}
\mathcal{L}_{T_{i,j,n}^{\rm dp}}(s)=\frac{\mathcal{L}_{T_{i,j,n}^{\rm dp,se}}(s)(1-\rho)s}{\Lambda_n^{\rm A}\mathcal{L}_{T_{i,j,n}^{\rm dp,se}}(s)\Lambda_n^{\rm A}+s},
\end{equation}  
where $\mathcal{L}_{T_{i,j,n}^{\rm dp,se}}(s)$ is defined as the Laplace transform of processing delay in the FAP, $\rho$ equals $\Lambda_n^{\rm A}\mathbb{E}[T_{i,j,n}^{\rm dp,se}]$. Note that the PDF of the service rate in the FAP is given in \eqref{equ_pdf_serv_time_fap}. Consequently, the $\mathcal{L}_{T_{i,j,n}^{\rm dp,se}}(s)$ can be calculated as follows:
\begin{equation}
\label{appd_lemma_lt_dd_tc_time_1}
\begin{aligned}
\mathcal{L}_{T_{i,j,n}^{\rm dp,se}}(s)&=\int_0^{\infty}{\exp(-sx)f_{\mu_n^{\rm A}}(x)}{\rm d}x\\
&=\frac{\mu_n^{\rm dd}\mu_n^{\rm cp}}{(s+\mu_n^{\rm dd})(s+\mu_n^{\rm cp})}.
\end{aligned}
\end{equation}
Moreover, the term $\rho$ can be expressed as:
\begin{equation}
\label{appd_lemma_lt_dd_tc_time_2}
\rho=\Lambda_n^{\rm A}\left(\frac{1}{\mu_n^{\rm dd}}+\frac{1}{\mu_n^{\rm cp}} \right).
\end{equation}
By incorporating \eqref{appd_lemma_lt_dd_tc_time_1} and \eqref{appd_lemma_lt_dd_tc_time_2} into \eqref{eq_dd_cp_lap_def}, we obtain the Laplace transform of the entire delay on the DD and task computing in the FAP $X_n^{\rm A}$ as follows:
\begin{equation}
\mathcal{L}_{T_{i,j,n}^{\rm dp}}(s)=\frac{s[\mu_n^{\rm dd}\mu_n^{\rm cp}-\Lambda_n^{\rm A}(\mu_n^{\rm dd}+\mu_n^{\rm cp})]}{(s-\Lambda_n^{\rm A})(s+\mu_n^{\rm dd})(s+\mu_n^{\rm cp})\!+\!\Lambda_n^{\rm A}\mu_n^{\rm dd}\mu_n^{\rm cp}}.
\end{equation} 
Then by calculating its inverse Laplace transform in Matlab, i.e., $\mathcal{L}^{-1}_{T_{i,j,n}^{\rm dp,se}}(t)=\mathcal{L}^{-1}\left[\mathcal{L}_{T_{i,j,n}^{\rm dp,se}}(s)\right]$, the result of $f_{T_{i,j,n}^{\rm dp}}(t)$ in Lemma \ref{lemma_lt_dd_tc_time} can be yielded. It is worth mentioning that we can prove that the integration $\int_0^{\infty}f_{T_{i,j,n}^{\rm dp}}(t)$ equals $1$, which satisfies the fundamental property of the PDF.

\section{Proof of Theorem \ref{theo_step}}
\label{appd_theo_sdcp}
The STEP $\Xi_{i,j,n}$ of the UE $Y_{i,j}$ associated with the FN $X_j^{\rm N}$ and the FAP $X_n^{\rm A}$ can be expressed as follows according to its definition:
\begin{equation}
\begin{aligned}
\Xi_{i,j,n}&=\mathbb{P}\left[T_{i,j,n}^{\rm tot,\xi}<\varrho\right]\\
&=\mathbb{P}\left[T_{i,j}^{\rm \xi,ut}+T_{i,j}^{\rm \xi,dc}+T^{\rm bh}_{i,j}+T^{\rm dp}_{i,j,n}<\varrho\right].
\end{aligned}
\end{equation}
For analytical tractability, we adopt the average ergodic rate to calculate the uplink transmission latency \cite{hu2019density}, which can be calculated by:
\begin{equation}
\label{uplink_rate}
\bar{C}_{i,j}^{\rm ut}=\int_0^{\infty}\log_2(1+\tau)\widetilde{\Delta}_{i,j}(\tau){\rm d}\tau.
\end{equation}
Combining with the result in Proposition \ref{prop_stp_appx},  the average uplink transmission latency $\bar{T}_{i,j}^{\rm \xi,ut}$ can be calculated as given in \eqref{equ_avg_up_trans}. Therefore, the STEP $\Xi_{i,j,n}$ can be transformed as:
\begin{equation}
\begin{aligned}
\Xi_{i,j,n}&=\mathbb{P}\left[T_{i,j}^{\rm \xi,dc}+T^{\rm dp}_{i,j,n}<\varrho-\bar{T}_{i,j}^{\rm \xi,ut}-T^{\rm dp}_{i,j,n}\right]\\
&=\mathbb{P}\left[\beta_{i,j}^{\xi}T_{i,j}^{\rm e,dc}+(1-\beta_{i,j}^{\xi})T_{i,j}^{\rm l,dc}+T^{\rm dp}_{i,j,n}<\varrho'_{\xi} \right],\\
\end{aligned}
\end{equation}
where $\varrho'_{\xi}$ denotes $\varrho-\bar{T}_{i,j}^{\rm \xi,ut}-T^{\rm dp}_{i,j,n}$ for denotational simplicity. Moreover, $\beta_{i,j}^{\xi}$ equals $1$ when $\xi={\rm e}$, equals $0$ when $\xi={\rm l}$, and equals $\beta_{i,j}$ otherwise. Note that both of the $T_{i,j}^{\rm e,dc}$ and $T_{i,j}^{\rm l,dc}$ are modelled following the M/M/1 model, thus $T_{i,j}^{\rm e,dc}$ and $T_{i,j}^{\rm l,dc}$ both follow the exponential distribution with parameters being $\sigma_{j}^{\rm N}=\mu_{j}^{\rm n,c}-\Lambda_{j}^{\rm N}$ and $\sigma_{i,j}^{\rm U}=\mu_{i,j}^{\rm u,c}-\Lambda_{i,j}$, respectively \cite{}, where $\Lambda_{j}^{\rm N}$ is defined as the task arriving rate at the FN $X_j^{\rm N}=M_j^{\rm U}\Lambda_{i,j}$. Accordingly, the STEP can be calculated as follows when $\beta_{i,j}^{\xi}\neq 0$:
\begin{equation}
\label{eq_appd_step_1}
\begin{aligned}
&\Xi_{i,j,n}=\mathbb{P}\left\{T_{i,j}^{\rm e,dc}<\frac{1}{\beta_{i,j}^{\xi}}\left[\varrho'_{\xi}-(1-\beta_{i,j}^{\xi})T_{i,j}^{\rm l,dc}-T^{\rm dp}_{i,j,n}\right] \right\}\\
&\!=\!\mathbb{E}_{T_{i,j}^{\rm l,dc},T^{\rm dp}_{i,j,n}}\left[
1-e^{-\sigma_{j}^{\rm N}\left(
	\frac{\varrho'_{\xi}}{\beta_{i,j}^{\xi}}-\frac{1-\beta_{i,j}^{\xi}}{\beta_{i,j}^{\xi}}T_{i,j}^{\rm l,dc}-\frac{T^{\rm dp}_{i,j,n}}{\beta_{i,j}^{\xi}}
	\right)
	}\right]\\
	&\!=\!\iint\limits_{\mathbb{L}}
	\left[
	1\!-\!e^{\!-\sigma_{j}^{\rm N}\left(
	\frac{\varrho'_{\xi}}{\beta_{i,j}^{\xi}}-\frac{1\!-\!\beta_{i,j}^{\xi}}{\beta_{i,j}^{\xi}}x-\frac{t}{\beta_{i,j}^{\xi}}
	\right)
	}\right]
	f_{T_{i,j}^{\rm l,dc}}(x)f_{T^{\rm dp}_{i,j,n}}(t)
	{\rm d}x{\rm d}t \\
	\end{aligned}
\end{equation}\\
where $\mathbb{L}$ is a two-dimensional integration region following $\mathbb{L}=\left\{(x,t):\frac{1-\beta_{i,j}^{\xi}}{\beta_{i,j}^{\xi}}x+\frac{1}{\beta_{i,j}^{\xi}}t\leq\frac{\varrho'_{\xi}}{\beta_{i,j}^{\xi}}\right\}$. This two-dimensional integration can be further translated as:
\begin{equation}
\label{equ_appd_step_2}
\begin{aligned}
&	\Xi_{i,j,n}\!=\!\int^{\varrho'_{\xi}}_{0}
	\int^{\frac{\varrho'_{\xi}\!-\!t}{1\!-\!\beta_{i,j}^{\xi}}}_{0}
	\left[
	1\!-\!e^{-\sigma_{j}^{\rm N}\left(
	\frac{\varrho'}{\beta_{i,j}^{\xi}}-\frac{1-\beta_{i,j}^{\xi}}{\beta_{i,j}^{\xi}}x-\frac{t}{\beta_{i,j}^{\xi}}
	\right)
	}\right]
	f_{T_{i,j}^{\rm l,dc}}(x)f_{T^{\rm dp}_{i,j,n}}(t)
	{\rm d}x{\rm d}t\\
	&\!=\! \int^{\varrho'_{\xi}}_{0}
	\int^{\frac{\varrho'_{\xi}-t}{1-\beta_{i,j}^{\xi}}}_{0}
	f_{T_{i,j}^{\rm l,dc}}(x)f_{T^{\rm dp}_{i,j,n}}(t)
	{\rm d}x{\rm d}t-
	\int^{\varrho'_{\xi}}_{0}
	\int^{\frac{\varrho'_{\xi}\!-\!t}{1\!-\!\beta_{i,j}^{\xi}}}_{0}
	e^{-\!\sigma_{j}^{\rm N}\left(
	\frac{\varrho'_{\xi}}{\beta_{i,j}^{\xi}}\!-\!\frac{(1-\beta_{i,j}^{\xi})}{\beta_{i,j}^{\xi}}x\!-\!\frac{t}{\beta_{i,j}^{\xi}}
	\right)
	}
	f_{T_{i,j}^{\rm l,dc}}(x)f_{T^{\rm dp}_{i,j,n}}(t)
	{\rm d}x{\rm d}t.
\end{aligned}
\end{equation}
The first term in \eqref{equ_appd_step_2} can be calculated as:
\begin{equation}
\label{eq_appd_step_3}
\begin{aligned}
	   &\int_{0}^{\varrho'_{\xi}} \int_{0}^{\frac{\varrho'_{\xi}-t}{1-\beta_{i,j}^{\xi}}}
       f_{T_{i,j}^{\rm l,dc}}(x)f_{T^{\rm dp}_{i,j,n}}(t)
       {\rm d}x{\rm d}t =-\int^{\varrho'_{\xi}}_{0}\left[e^{-\sigma_{i,j}^{\rm U}\frac{\varrho'_{\xi}-t}{1-\beta_{i,j}^{\xi}}}-1\right]
           f_{T^{\rm dp}_{i,j,n}}(t){\rm d}t\\
       &\overset{(a)}=\frac{\mu_n^{\rm dd}\mu_n^{\rm cp} \left(1-\rho\right)  }{\eta_n}
         \left\{e^{-\frac{\sigma_{i,j}^{\rm U} \varrho'_{\xi}}{1-\beta_{i,j}^{\xi}}}
         \left[\frac{2(1-\beta_{i,j}^{\xi})}{2\sigma_{i,j}^{\rm U}+(1-\beta_{i,j}^{\xi})\varsigma_n^{\rm A}}
         \left(e^{\frac{2\sigma_{i,j}^{\rm U}+(1-\beta_{i,j}^{\xi}){\varsigma_n^{\rm A}}}{2(1-\beta_{i,j}^{\xi})}\varrho'_{\xi}}-1\right)   
         -\right.\right.\\
         &\left.\left.\frac{2(1-\beta_{i,j}^{\xi})}{2\sigma_{i,j}^{\rm U}+(1-\beta_{i,j}^{\xi}){\varsigma_n^{\rm A'}}}\left(e^{\frac{2\sigma_{i,j}^{\rm U}+(1-\beta_{i,j}^{\xi}){\varsigma_n^{\rm A'}}}{2(1-\beta_{i,j}^{\xi})}\varrho'_{\xi}}\!-\!1\right)                  
         \right]
        -\left(
         \frac{2}{\varsigma_n^{\rm A}}e^{\frac{\varsigma_n^{\rm A}}{2} \varrho'_{\xi}}
         -\frac{2}{\varsigma_n^{\rm A'}}e^{\frac{\varsigma_n^{\rm A'}} {2} \varrho'_{\xi}}
         -\frac{2}{\varsigma_n^{\rm A}}+\frac{2}{\varsigma_n^{\rm A'}}
         \right)
         \right\}\\
\end{aligned}
\end{equation}
where step (a) is obtained by incorporating the results of $f_{T^{\rm dp}_{i,j,n}}(t)$ given in \eqref{eq_dp_delay_distr} into \eqref{eq_appd_step_3}.
The second term in \eqref{equ_appd_step_2} can be calculated as follows:
\begin{equation}
	\label{eq_appd_step_4}
	\begin{aligned}
	&\int^{\varrho'_{\xi}}_{0}
	\int^{\frac{\varrho'_{\xi}\!-\!t}{1\!-\!\beta_{i,j}^{\xi}}}_{0}
	e^{-\!\sigma_{j}^{\rm N}\left(
	\frac{\varrho'_{\xi}}{\beta_{i,j}^{\xi}}\!-\!\frac{(1-\beta_{i,j}^{\xi})}{\beta_{i,j}^{\xi}}x\!-\!\frac{t}{\beta_{i,j}^{\xi}}
	\right)
	}
	f_{T_{i,j}^{\rm l,dc}}(x)f_{T^{\rm dp}_{i,j,n}}(t)
	{\rm d}x{\rm d}t\\
	&=\int^{\varrho'_{\xi}}_{0}
	  \frac{\beta_{i,j}^{\xi} \sigma_{i,j}^{\rm U}}{\left(\beta_{i,j}^{\xi}-1\right) \sigma_{j}^{\rm N}- \beta_{i,j}^{\xi} \sigma_{i,j}^{\rm U}}
	  \left[ e^{\left(\varrho'_{\xi}-t\right)\left(\frac{\sigma_{j}^{\rm N}}{\beta_{i,j}^{\xi}}-\frac{\sigma_{i,j}^{\rm U}}{1-\beta_{i,j}^{\xi}}\right)}  -1\right]\
	  e^{-\frac{\varrho'_{\xi} \sigma_{j}^{\rm N}}{\beta_{i,j}^{\xi}}}
	  e^{\frac{\sigma_{j}^{\rm N}}{\beta_{i,j}^{\xi}} t}
	  f_{T^{\rm dp}_{i,j,n}}(t)
	  {\rm d}t\\
	 &=-\frac{\beta_{i,j}^{\xi} \sigma_{i,j}^{\rm U} \mu_n^{\rm dd}\mu_n^{\rm cp} \left(1-\rho\right)}
	   {\eta_n \left[\left(1-\beta_{i,j}^{\xi}\right)\sigma_{j}^{\rm N}-\beta_{i,j}^{\xi} \sigma_{i,j}^{\rm U}\right]}
	   \left\{
	         e^{-\frac{\sigma_{i,j}^{\rm U} \varrho'_{\xi}}{1-\beta_{i,j}^{\xi}}}
	         \frac{2\left(1-\beta_{i,j}^{\xi}\right)}
	         {2 \sigma_{i,j}^{\rm U} + \left(1-\beta_{i,j}^{\xi}\right) \varsigma_n^{\rm A}}
	         \left[e^{\frac{2 \sigma_{i,j}^{\rm U} + \left(1-\beta_{i,j}^{\xi}\right) \varsigma_n^{\rm A'}}{2\left(1-\beta_{i,j}^{\xi}\right)} \varrho'_{\xi}}-1  \right] \right.\\	         
	&\left. -e^{-\frac{\sigma_{j}^{\rm N} \varrho'_{\xi}}{\beta_{i,j}^{\xi}}}
	         \frac{2\beta_{i,j}^{\xi}}
	         {2 \sigma_{j}^{\rm N} + \beta_{i,j}^{\xi} \varsigma_n^{\rm A}}
	         \left[e^{\frac{2 \sigma_{j}^{\rm N} + \beta_{i,j}^{\xi} \varsigma_n^{\rm A}}{2\beta_{i,j}^{\xi}} \varrho'_{\xi}}\!-\!1  \right] 
	        -e^{-\frac{\sigma_{i,j}^{\rm U} \varrho'_{\xi}}{1-\beta_{i,j}^{\xi}}}
	         \frac{2\left(1-\beta_{i,j}^{\xi}\right)}
	         {2 \sigma_{i,j}^{\rm U}\! +\! \left(1\!-\!\beta_{i,j}^{\xi}\right) {\varsigma_n^{\rm A'}}}
	         \left[e^{\frac{2 \sigma_{i,j}^{\rm U} \!+\! \left(1-\beta_{i,j}^{\xi}\right) {\varsigma_n^{\rm A'}}}{2\left(1-\beta_{i,j}^{\xi}\right)} \varrho'_{\xi}}\!-\!1  \right] \right.\\ 
	&\left. +e^{-\frac{\sigma_{j}^{\rm N} \varrho'_{\xi}}{\beta_{i,j}^{\xi}}} 	
	         \frac{2\beta_{i,j}^{\xi}}
	         {2 \sigma_{j}^{\rm N} + \beta_{i,j}^{\xi} {\varsigma_n^{\rm A'}}}
	         \left[e^{\frac{2 \sigma_{j}^{\rm N} + \beta_{i,j}^{\xi} {\varsigma_n^{\rm A'}}}{2\beta_{i,j}^{\xi}} \varrho'_{\xi}}-1  \right] 
	   \right\} .
	\end{aligned}
\end{equation}
By incorporating \eqref{eq_appd_step_3} and \eqref{eq_appd_step_4} into \eqref{equ_appd_step_2}, the STEP $\Xi_{i,j,n}$ for the DC mode $\xi$, $\xi\in\{\rm{h,e}\}$ can be calculated, and by some symbolic transformations, the result in \eqref{equ_clo_step} is obtained.

For the local compression mode, i.e., $\xi={\rm l}$ and $\beta_{i,j}=0$, the STEP $\Xi_{i,j,n}$ can be calculated as follows:
\begin{equation}
\begin{aligned}
\Xi_{i,j,n}&\!=\!\mathbb{P}\left[T_{i,j}^{\rm l,dc}+T^{\rm dp}_{i,j,n}<\varrho-\bar{T}_{i,j}^{\rm l,ut}-T^{\rm bh}_{i,j,n}\right]\\
&\!=\!\mathbb{P}\left[T_{i,j}^{\rm l,dc}<\varrho'_{\rm l}-T^{\rm dp}_{i,j,n} \right],\\
&\!=\!\mathbb{E}_{T^{\rm dp}_{i,j,n}}\left\{1-\exp\left[-\sigma_{i,j}^{\rm U}
\left(   \varrho'_{\rm l}-t  \right)
\right] \right\}\\
&\!=\!\int_{0}^{\varrho'_{\rm l}} \left\{ 1-\exp\left[-\sigma_{i,j}^{\rm U}
\left(   \varrho'_{\rm l}-t  \right)
\right]\right\}
f_{T^{\rm dp}_{i,j,n}}(t){\rm d}t\\
&=
    \frac{\mu_n^{\rm dd}\mu_n^{\rm cp} \left(1-\rho\right)  }{\eta_n}
    \left[
    \mathcal{A}\left(\varsigma_n^{\rm A},1,{\sigma_{i,j}^{\rm U}}\right)
    \!-\!\mathcal{A}\left({\varsigma_n^{\rm A'}},1,{\sigma_{i,j}^{\rm U}}\right)-\left(\frac{2}{\varsigma_n^{\rm A}}e^{\frac{\varsigma_n^{\rm A}}{2}\varrho'_{\rm l}}
    \!-\!\frac{2}{\varsigma_n^{\rm A'}}e^{\frac{\varsigma_n^{\rm A'}}{2}\varrho'_{\rm l}}
    \!-\!\frac{2}{\varsigma_n^{\rm A}}
    \!+\!\frac{2}{\varsigma_n^{\rm A'}}\right)
    \right].
\end{aligned}
\end{equation}
Note that the result for the compression mode is the same as the results for the other two compression modes as given in \eqref{equ_clo_step} with $\beta_{i,j}=0$. Consequently, the STEP $\Xi_{i,j,n}$ can be generally expressed as the result given in Theorem \ref{theo_step}.

\bibliographystyle{IEEEtran}
\bibliography{On_the_Performance_of_Data_Compression_in_Clustered_Fog_Radio_Access_Networks}

\end{document}